\documentclass[11pt]{article}
\usepackage[a4paper, total={6in, 8in}]{geometry}
\usepackage{enumerate}
\usepackage{amsmath}
\usepackage{amssymb,latexsym}
\usepackage{amsthm}
\usepackage{color}
\usepackage{cancel}
\usepackage{graphicx}
\usepackage{cite}
\usepackage{changes}
\usepackage{lineno}
\usepackage{hyperref}
\usepackage{comment}
\usepackage{amscd}

\DeclareMathOperator{\C}{\mathcal{C}}

\DeclareMathOperator{\supp}{supp}

\newtheorem{theorem}{Theorem}[section]
\newtheorem{problem}[theorem]{Problem}
\newtheorem{lemma}[theorem]{Lemma}

\newtheorem{definition}[theorem]{Definition}
\newtheorem{proposition}[theorem]{Proposition}

\newtheorem{remark}[theorem]{Remark}

\newtheorem{construction}[theorem]{Construction}

\newcommand{\fqn}{\mathbb{F}_{q^n}}

\newcommand{\F}{{\mathbb F}}

\newcommand{\GL}{\hbox{{\rm GL}}}

\newcommand{\fq}{{\mathbb F}_{q}}

\newcommand{\la}{\langle}
\newcommand{\ra}{\rangle}

\newcommand{\PG}{\mathrm{PG}}

\title{Maximum weight codewords of a linear rank metric code}
\author{Olga Polverino, Paolo Santonastaso and Ferdinando Zullo}
\date{ }

\begin{document}
\maketitle

\begin{abstract}
Let $\mathcal{C}\subseteq \mathbb{F}_{q^m}^n$ be an $\mathbb{F}_{q^m}$-linear non-degenerate rank metric code with dimension $k$.
In this paper we investigate the problem of determining the number $M(\C)$ of codewords in $\C$ with maximum weight, that is $\min\{m,n\}$, and to characterize those with the maximum and the minimum values of $M(\C)$.
\end{abstract}


\section{Introduction}

Rank metric codes gained a lot of attention in the last decades due to their numerous applications and their connections with interesting mathematical objects. As a matter of fact, Silva, K\"otter and Kschischang in  \cite{silva2008rank} proposed the use of rank metric codes in linear random network coding. However, the origin of rank metric codes dates back to Delsarte \cite{de78} in  1978, some years later they were rediscovered by Gabidulin in \cite{ga85a} and Roth in \cite{roth1991maximum}.
Since then applications in criss-cross error corrections, cryptography and network coding arose, see e.g. \cite{bartz2022rank}.
Rank metric codes are also related to well-studied algebraic and combinatorial objects, such as semifields \cite{sheekey2016new}, linear sets in finite geometry \cite{polverino2020connections}, tensorial algebras \cite{byrne2019tensor},  skew algebras \cite{ACLN20+,elmaazouz2021}, $q$-analog of matroids \cite{gorla2020rank} and many more, see \cite{gorla2021rank} and \cite{sheekeysurvey}. 

In this paper we will be mostly concentrated in the case of linear codes, that is $\F_{q^m}$-subspaces of $\F_{q^m}^n$. Here, we equip $\F_{q^m}^n$ with the rank distance which is defined as follows: for any $u=(u_1,\ldots,u_n),v=(v_1,\ldots,v_n) \in \F_{q^m}^n$ then
\[d(u,v)=\dim_{\fq}(\langle u_1-v_1,\ldots,u_n-v_n \rangle_{\fq}).\]
Our aim is to give information on the weight distribution of $\mathbb{F}_{q^m}$-linear non-degenerate rank metric codes, that is $\F_{q^m}$-linear rank metric codes which cannot be embedded in another smaller space preserving its weight distribution.
For some classes of rank metric codes, the weight distribution is well-know, such as for MRD codes or classes of few weight codes, but in general very few is known.
In the Hamming metric, Ball and Blokhuis in \cite{ball2013bound} studied conditions on the code that guarantee that it contains a codeword with weight equals to the length of the code, by using the geometry of $t$-fold blocking sets in affine spaces.
For non-degenerate rank metric codes, in \cite[Proposition 3.11]{alfarano2022linear} the authors proved that there always exists a codeword of maximum weight $\min\{m,n\}$, allowing them to obtain a concise proof of the characterization of the optimal $\F_{q^m}$-linear anticodes originally proved in \cite[Theorem 18]{ravagnani2016generalized}.
The maximum weight codewords seems also interesting in connection with the rank metric version of the Critical problem by Crapo and Rota (cf. \cite{semgiani,gianibyrne} and see also \cite{gruica2022rank}), and due to the connection with $q$-polymatroids, see \cite{gluesing2022q}.

Let $\C$ be a $\F_{q^m}$-linear non-degenerate rank metric code in $\F_{q^m}^n$
of dimension $k$ and define $M(\C)$ as the number of codewords in $\C$ with weight $\min\{m,n\}$.
In this paper we investigate the following two problems:

\begin{problem}\label{probl1}
To determine upper and lower bounds on $M(\C)$.
\end{problem}

\begin{problem}\label{probl1}
To characterize the extremal cases in the obtained bounds on $M(\C)$.
\end{problem}

The main tools used in this paper are from combinatorics: we use the projective version of systems, namely linear sets, which are point sets in projective spaces. 
Then we use old and new bounds regarding the size and the weight distribution of linear sets, to obtain the desired bounds. 
As a consequence, once five values among $m,n,k,q, M(\C)$ and the second maximum weight are known, then we are able to determine the remaining one.
Then we provide examples and characterization results for the equality case in the obtained bounds, by making use of duality theory of linear sets, new and old constructions.
To give an idea of our results, we were able to prove that for the $2$-dimensional case, $M(\C)$ is minimum if and only if $\C$ or its geometric dual is an MRD code, which we prove is extendable to higher dimension only in the case in which the length of the code is $mk/2$.

The paper is structured as follows. In Section \ref{sec:prel}, we describe definitions and results on rank metric codes and linear sets needed for our results. Section \ref{sec:bounds} deals with upper and lower bounds on $M(\C)$, where the discussion is divided in four parts, according to the dimension of the code and the relation between $n$ and $m$. In Section \ref{sec:lowerbounds} we analyze the case of equality in the lower bounds by detecting the geometry of these codes, which are strongly related to scattered linear sets and hence (in some cases) to MRD codes. Section \ref{sec:upperbounds} is devoted to the case of equality in the upper bounds: the geometry in this case is either related to canonical subgeometries or to linear sets with minimum size.
Finally, we conclude the paper by listing some open problems.

\section{Preliminaries}\label{sec:prel}

We start fixing the following notation. Let $p$ be a prime and let $h$ be a positive integer. We fix $q=p^h$ and denote by $\fq$ the finite field with $q$ elements. Moreover, if $m$ is a positive integer then we may consider the extension field $\F_{q^m}$ of degree $m$ over $\fq$. 
Recall that for the extension $\F_{q^m}/\fq$, the \textbf{trace} of an element $\alpha \in \F_{q^m}$ is defined as
$$ \mathrm{Tr}_{q^m/q}(\alpha):= \sum_{i=0}^{m-1}\alpha^{q^i}.$$

\noindent We list some more notation which will be repeatedly used in this paper.

\begin{itemize}
    \item $V=V(k,q)$ denotes a $k$-dimensional $\fq$-vector space;
    \item $\langle U \rangle_{\F_{q}}$ denotes the $\fq$-span of $U$, with $U$ a subset of a vector space $V$;
    \item $\PG(k-1,q)$ denotes the projective Desarguesian space of dimension $k-1$ and order $q$;
    \item $\PG(V,\fq)$, with $V$ an $\fq$-vector space, denotes the projective space obtained by $V$;
    \item $\langle S \rangle$ denotes the span of the points in $S$, with $S$ a subset of $\PG(k-1,q)$;
    \item $\mathrm{GL}(k,q)$ denotes the general linear group;
    \item $\mathrm{\Gamma L}(k,q)$ denotes the general semilinear group;
    \item  $\mathrm{colsp}(A)$ is the $\fq$-span of the columns of a matrix $A$.
\end{itemize}

\subsection{Rank metric codes and $q$-systems}

\subsubsection{Generalities on rank metric codes}

The rank (weight) $w(v)$ of a vector $v=(v_1,\ldots,v_n) \in \F_{q^m}^n$ is the dimension of the vector space generated over $\F_q$ by its entries, i.e, $w(v)=\dim_{\fq} (\langle v_1,\ldots, v_n\rangle_{\fq})$. 

A \textbf{(linear vector) rank metric code} $\C $ is an $\F_{q^m}$-subspace of $\F_{q^m}^n$ endowed with the rank distance defined as
\[
d(x,y)=w(x-y),
\]
where $x, y \in \F_{q^m}^n$. 

Let $\C \subseteq \F_{q^m}^n$ be an $\F_{q^m}$-linear rank metric code. We will write that $\C$ is an $[n,k,d]_{q^m/q}$ code (or $[n,k]_{q^m/q}$ code) if $k$ is the $\F_{q^m}$-dimension of $\C$ and $d$ is its minimum distance, that is 
\[
d=\min\{d(x,y) \colon x, y \in \C, x \neq y  \}.
\]
Moreover, we denote by $A_i(\C)$, or simply $A_i$, the number of codewords in $\C$ of weight $i \in \{0,\ldots,n\}$ and $(A_0,\ldots,A_n)$ is called the \textbf{weight distribution} of $\C$.

It is possible to prove a Singleton-like bound for a rank metric code, which for the case of $\F_{q^m}$-linear codes reads as follows.

\begin{theorem} \cite{de78}\label{th:singletonrank}
    Let $\C \subseteq \F_{q^m}^n$ be an $[n,k,d]_{q^m/q}$ code.
Then 
\begin{equation}\label{eq:boundgen}
mk \leq \max\{m,n\}(\min\{n,m\}-d+1).\end{equation}
\end{theorem}

A $[n,k,d]_{q^m/q}$ code is called \textbf{Maximum Rank Distance code} (or shortly \textbf{MRD code}) if its parameters attains the bound \eqref{eq:boundgen}.

Recall that the $q$-binomial coefficient of two integers $s$ and $t$ is
\[  {s \brack t}_q=\left\{ \begin{array}{lll} 0 & \text{if}\,\, s<0,\,\,\text{or}\,\,t<0,\,\, \text{or}\,\, t>s,\\ 
1 & \text{if}\,\, t=0\,\, \text{and}\,\, s\geq 0,\\
\displaystyle\prod_{i=1}^t \frac{q^{s-i+1}-1}{q^i-1} & \text{otherwise}. \end{array} \right.  \]

This number counts the number of $t$-dimensional $\fq$-subspaces of an $s$-dimensional $\fq$-vector space.

Delsarte in \cite{de78} and later Gabidulin in \cite{ga85a} determined MacWilliams identities for rank metric codes which yield in case of an MRD code to the determination of the weight distribution of MRD codes, see also \cite{ravagnani2016rank}.

\begin{theorem}\label{th:weightdistribution}
Let $\mathcal{C}$ be an MRD code with parameters $[n,k,d]_{q^m/q}$. Let $m'=\min\{m,n\}$ and $n'=\max\{m,n\}$.
Then
\[ A_{d+\ell}={m'\brack d+\ell}_q \sum_{t=0}^\ell (-1)^{t-\ell}{\ell+d \brack \ell-t}_q q^{\binom{\ell-t}{2}}(q^{n'(t+1)}-1) \]
for any $\ell \in \{0,1,\ldots,n'-d\}$.
\end{theorem}

Another important notion is the rank support of a codeword.
Let $\Gamma=(\gamma_1,\ldots,\gamma_m)$ be an ordered $\fq$-basis of $\F_{q^m}$. For any vector $x=(x_1, \ldots ,x_n) \in \F_{q^m}^n$ define the matrix $\Gamma(x)\in \F_{q}^{n \times m}$, where
$$x_{i} = \sum_{j=1}^m \Gamma (x)_{ij}\gamma_j, \qquad \mbox{ for all } i \in \{1,\ldots,n\},$$
that is $\Gamma(x)$ is the matrix expansion of the vector $x$ with respect to the $\Gamma$ of $\F_{q^m}$ and this clearly preserves its rank, i.e. $w(x)=\mathrm{rk}(\Gamma(x))$.

\begin{definition}
Let $x=(x_1,\dots, x_n)\in\F_{q^m}^{n}$ and  $\Gamma=(\gamma_1,\ldots,\gamma_m)$ an order $\fq$-basis of $\F_{q^m}$.
The \textbf{rank support} of $x$ is defined as the column span of $\Gamma(x)$:
$$\supp(x)=\mathrm{colsp}(\Gamma(x)) \subseteq \fq^n.$$
\end{definition}

As proved in \cite[Proposition 2.1]{alfarano2022linear}, the support does not depend on the choice of $\Gamma$ and we can talk about the support of a vector without mentioning $\Gamma$.

For more details we refer to \cite{martinez2019theory}.

\subsubsection{Geometry of rank metric codes}\label{sec:georankmetric}

Now, we recall the definition of equivalence between rank metric codes in $\F_{q^m}^n$. An $\F_{q^m}$-linear isometry $\phi$ of $\F_{q^m}^n$ is an $\F_{q^m}$-linear map of $\F_{q^m}^{n}$ that preserves the distance, i.e. $w(x)=w(\phi(x))$, for every $x \in  \F_{q^m}^{n}$, or equivalently $d(x,y)=d(\phi(x),\phi(y))$, for every $x,y \in  \F_{q^m}^{n}$.
It has been proved that the group of $\F_{q^m}$-linear isometries of $\F_{q^m}^n$ equipped with rank distance is generated by the (nonzero) scalar
multiplications of $\F_{q^m}$ and the linear group $\mathrm{GL}(n,\F_q)$, see e.g. \cite{berger2003isometries}. Therefore, we say that two rank metric codes $\C,\C' \subseteq \F_{q^m}^n$ are \textbf{(linearly) equivalent} if there exists an isometry $\phi$ such that $\phi(\C)=\C'$.
Clearly, when studying equivalence of $[n, k]_{q^m/q}$ codes the
action of $\F_{q^m}^*$ is trivial. This means that two $[n,k]_{q^m/q}$ codes $\C$ and $\C'$ are equivalent if and only if
there exists $A \in \mathrm{GL}(n,q)$ such that
$\C'=\C A=\{vA : v \in \C\}$. 
Most of the codes we will consider are \emph{non-degenerate}.

 \begin{definition}
An $[n,k]_{q^m/q}$ rank metric code $\C$ is said to be \textbf{non-degenerate} if the columns of any generator matrix of $\C$ are $\fq$-linearly independent.
We denote the set of equivalence classes of $[n,k,d]_{q^m/q}$ non-degenerate rank metric codes by $\mathfrak{C}[n,k,d]_{q^m/q}$.
\end{definition}

The geometric counterpart of rank metric are the systems. 

 \begin{definition}
An $[n,k,d]_{q^m/q}$ \textbf{system} $U$ is an $\F_q$-subspace of $\F_{q^m}^k$ of dimension $n$, such that
$ \langle U \rangle_{\F_{q^m}}=\F_{q^m}^k$ and

$$ d=n-\max\left\{\dim_{\F_q}(U\cap H) \mid H \textnormal{ is an $\F_{q^m}$-hyperplane of }\F_{q^m}^k\right\}.$$

Moreover, two $[n,k,d]_{q^m/q}$ systems $U$ and $U'$ are \textbf{equivalent} if there exists an $\F_{q^m}$-isomorphism $\varphi\in\GL(k,\F_{q^m})$ such that
$$ \varphi(U) = U'.$$
We denote the set of equivalence classes of $[n,k,d]_{q^m/q}$ systems by $\mathfrak{U}[n,k,d]_{q^m/q}$.
\end{definition}

The following result allows us to establihs a correspondence between rank metric codes and systems.

\begin{theorem}\cite{Randrianarisoa2020ageometric} \label{th:connection}
Let $\C$ be a non-degenerate $[n,k,d]_{q^m/q}$ rank metric code and let $G$ be a generator matrix.
Let $U \subseteq \F_{q^m}^k$ be the $\F_q$-span of the columns of $G$.
The rank weight of an element $x G \in \C$, with $x=(x_1,\ldots,x_k) \in \F_{q^m}^k$ is
\begin{equation}\label{eq:relweight}
w(x G) = n - \dim_{\fq}(U \cap x^{\perp}),\end{equation}
where $x^{\perp}=\{y=(y_1,\ldots,y_k) \in \F_{q^m}^k \colon \sum_{i=1}^k x_iy_i =0\}.$ In particular,
\begin{equation} \label{eq:distancedesign}
d=n - \max\left\{ \dim_{\fq}(U \cap H)  \colon H\mbox{ is an } \F_{q^m}\mbox{-hyperplane of }\F_{q^m}^k  \right\}.
\end{equation}
\end{theorem}

Thanks to the above theorem, we have a complete correspondence between nondegenerate $[n,k,d]_{q^m/q}$ codes and $[n,k,d]_{q^m/q}$ systems.

\begin{theorem}\label{thm:1-1correspondence}
There is a one-to-one  correspondence  between  equivalence  classes of nondegenerate $[n,k,d]_{q^m/q}$ codes and equivalence classes of $[n,k,d]_{q^m/q}$ systems.
\end{theorem}

The correspondence can be formalized by the following two maps
\begin{align*}
    \Psi :  \mathfrak{C}[n,k,d]_{q^m/q} &\to\mathfrak{U}[n,k,d]_{q^m/q} \\
    \Phi : \mathfrak{U}[n,k,d]_{q^m/q} &\to \mathfrak{C}[n,k,d]_{q^m/q},
\end{align*}
which act as follows.
Let $[\C]\in\mathfrak{C}[n,k,d]_{q^m/q}$ and $G$ be a generator matrix for $\C$. Then $\Psi([\C])$ is the equivalence class of $[n,k,d]_{q^m/q}$ systems $[U]$, where $U$ is the $\F_q$-span of the columns of $G$. In this case $U$ is also called a \textbf{system associated with} $\C$. Viceversa, given $[U]\in\mathfrak{U}[n,k,d]_{q^m/q}$. Define $G$ as the matrix whose columns are an $\F_q$-basis of $U$ and let $\C$ be the code generated by $G$. Then $\Phi([U])$ is the equivalence class of the $[n,k,d]_{q^m/q}$ codes $[\C]$. $\C$ is also called a \textbf{code associated with} $U$. $\Psi$ and $\Phi$ are well-posed and they are inverse of each other. See also \cite{alfarano2022linear}.

An important code, whose definition arises naturally from the geometric view, is the \textbf{simplex code}, which has been defined in \cite{alfarano2022linear} as any non-degenerate $[mk,k]_{q^m/q}$ code, that is having as an associated system $\F_{q^m}^k$ seen as an $\fq$-vector space of dimension $mk$.

Generalized rank weights have been introduced several times with different definition, see e.g. \cite{jurrius2017defining}, and they have been used also as a tool for the inequivalence of families of codes as was done in \cite{bartoli2022new}.

In this paper we will deal with the definition given in \cite{Randrianarisoa2020ageometric} and more precisely to the equivalent one given in \cite[Theorem 3.14]{alfarano2022linear}, directly connected with the systems.

\begin{definition}
Let $\C$ be a non-degenerate $[n,k,d]_{q^m/q}$ rank metric code and let $U$ be an associated system.
For any $r \in \{1,\ldots,k\}$, the \textbf{$r$-th generalized rank weight} is
\begin{equation}\label{eq:defgenrankweight}
d_r^{\mathrm{rk}}(\C)=n - \max\left\{ \dim_{\fq}(U \cap H)  \colon H\mbox{ is an } \F_{q^m}\mbox{-subspace of codim. $r$ of  }\F_{q^m}^k  \right\}.
\end{equation}
\end{definition}

Note that when $r=1$, in the above defintion we obtain the minimum distance.

In the next, we will recall how the support of a codeword is related to the intersections with a system associated with the code.

Let $G \in \mathbb{F}_{q^m}^{k\times n}$ such that its columns are $\fq$-linearly independent and let $U$ be the $\mathbb{F}_q$-span of the columns of $G$. 
Define the map
\[
\begin{array}{rccl}
\psi_{G}:& \fq^{n}& \longrightarrow &U \\
&\lambda & \longmapsto & \lambda G^\top,\end{array}\]
which turns out to be an $\fq$-linear isomorphism.

\begin{theorem}\cite[Theorem 3.1]{neri2021geometry}\label{th:suppogeom}
Let $\C$ be a non-degenerate $[n,k]_{q^m/q}$ code with generator matrix $G \in \F_{q^m}^{k\times n}$ and let $U$ be the $\F_q$-span of the columns of $G$. Then, for every $x\in \mathbb{F}_{q^m}^k$
$$\psi_{G}^{-1}(U \cap x^\perp)=\supp(vG)^\perp,$$
where $\supp(vG)^\perp$ si the orthogonal complement of $\supp(vG)$ with respect to the standard scalar product in $\fq^n$.
\end{theorem}

\subsection{Linear sets}

In this paper we will often use and look to the systems projectively via the notion of linear sets.
Let $V$ be a $k$-dimensional vector space over $\F_{q^m}$ and let $\Lambda=\PG(V,\F_{q^m})=\PG(k-1,q^m)$.
Recall that, if $U$ is an $\fq$-subspace of $V$ of dimension $n$, then the set of points
\[ L_U=\{\la { u} \ra_{\mathbb{F}_{q^m}} : { u}\in U\setminus \{{ 0} \}\}\subseteq \Lambda \]
is said to be an $\fq$-\textbf{linear set of rank $n$}.
Let $\Omega=\PG(W,\F_{q^m})$ be a projective subspace of $\Lambda$. The \textbf{weight of $\Omega$} in $L_U$ is defined as 
\[ w_{L_U}(\Omega)=\dim_{\fq}(U\cap W). \]
If $N_i$ denotes the number of points of $\Lambda$ having weight $i\in \{0,\ldots,n\}$  in $L_U$, the following relations hold:
\begin{equation}\label{eq:card}
    |L_U| \leq \frac{q^n-1}{q-1},
\end{equation}
\begin{equation}\label{eq:pesicard}
    |L_U| =N_1+\ldots+N_n,
\end{equation}
\begin{equation}\label{eq:pesivett}
    N_1+N_2(q+1)+\ldots+N_n(q^{n-1}+\ldots+q+1)=q^{n-1}+\ldots+q+1.
\end{equation}
Moreover, if $L_U\ne \emptyset$, then 
\begin{equation}\label{eq:cong1}
    |L_U|\equiv 1 \pmod{q},
\end{equation}
and if $\langle L_U \rangle =\PG(k-1,q^m)$ then 
\begin{equation}\label{eq:boundsubgeomcontainedpre}
    |L_U|\geq \frac{q^k-1}{q-1}.
\end{equation}

Note also that if $P_1,\ldots ,P_j \in L_U$ are independent points (i.e. $\langle P_1,\ldots,P_j \rangle$ is a $(j-1)$-dimensional subspace of $\PG(k-1,q^m)$), then 
\begin{equation} \label{eq:boundcomplementaryweight}
w_{L_U}(P_1) + \ldots+w_{L_U}(P_j) \leq n.
\end{equation}

\begin{remark}\label{rk:maximumweighpointscomplementaryweights}
In the case in which there exist $j$ independent points in $L_U$ such that in \eqref{eq:boundcomplementaryweight} the equality holds, then the maximum of the weight of the points in $L_U$ is 
\[ \max\{ w_{L_U}(P_i) \colon i \in\{1,\ldots,j\} \}. \]
\end{remark}

Furthermore, $L_U$ and $U$ are called \textbf{scattered} if $L_U$ has the maximum number $\frac{q^n-1}{q-1}$ of points, or equivalently, if all points of $L_U$ have weight one. 
\textbf{Canonical subgeometries} of $\PG(k-1,q^m)$ are defined as those $\fq$-linear set $L_U$ with rank $k$ spanning the entire space and they are examples of scattered $\fq$-linear sets.
Blokhuis and Lavrauw provided the following bound on the rank of a scattered liner set.

\begin{theorem}\cite{blokhuis2000scattered}\label{th:boundscattrank}
Let $L_U$ be a scattered $\fq$-linear set of rank $n$ in $\PG(k-1,q^m)$, then
\[ n\leq \frac{mk}2. \]
\end{theorem}

A scattered $\fq$-linear set of rank $km/2$ in $\PG(k-1,q^m)$ is said to be a \textbf{maximum scattered} and $U$ is said to be a maximum scattered $\fq$-subspace as well. 

A trivial lower bound on the number of points of a non-empty linear set $L_U$ is $|L_U|\geq 1$. It can be improved if some assumptions are added. 

\begin{theorem}[{\cite[Theorem 1.2]{debeule2019theminimumsize}} and {\cite[Lemma 2.2]{bonoli2005fq}}]\label{th:minsize}
 If $L_U$ is an $\fq$-linear set of rank $n$, with $1 < n \leq m$ on $\PG(1,q^m)$, and $L_U$ contains at least one point of weight $1$, then $|L_U| \geq q^{n-1} + 1$.
\end{theorem}

The following result deals with a family of examples having the minimum number of points, i.e. satisfying the bound of Theorem \ref{th:minsize}.

\begin{theorem}\cite[Theorem 2.7]{jena2021linear}\label{th:constructionVdV}
Let $\lambda\in \F_{q^m}$ be an element generating $\F_{q^m}$ and 
\[ L_U=\{ \langle (\alpha_0+\alpha_1 \lambda+\ldots+\alpha_{t_1-1}\lambda^{t_1-1},\beta_0+\beta_1 \lambda+\ldots+\beta_{t_2-1}\lambda^{t_2-1}) \rangle_{\F_{q^m}} \colon \alpha_i,\beta_i \in \fq,\,\]\[\text{not all zero},\, 1 \leq t_1,t_2, t_1+t_2 \leq m   \}, \]
where
\[U=\langle 1,\lambda, \ldots,\lambda^{t_1-1}\rangle_{\fq}\times \langle 1,\lambda, \ldots,\lambda^{t_2-1} \rangle_{\fq} .\]
Then $L_U$ is an $\fq$-linear set of $\PG(1,q^m)$ of rank $k=t_1+t_2$ with $q^{k-1}+1$ points.
Let $t_1\leq t_2$, then 
\begin{itemize}
\item the point $\langle (0,1)\rangle_{\fqn}$ has weight $t_2$;
\item there are $q^{t_2-t_1+1}$ points of weight $t_1$ different from $\langle (0,1)\rangle_{\fqn}$;
\item there are $q^{k-2i+1}-q^{k-2i-1}$ points of weight $i \in \{1,\ldots, t_1-1\}$.
\end{itemize}
\end{theorem}

Recently, extending the results in \cite{debeule2019theminimumsize}, in \cite{SamPaolo} the following lower bound on the size of a linear set has been proved.

\begin{theorem}\cite{SamPaolo,debeule2019theminimumsize}
 \label{thm:newBound}
Let $L_U$ be an $\fq$-linear set of rank $n$ in $\PG(k-1,q^m)$ spanning the whole space.
Suppose that there exists some $(r-1)$-space $\Omega$, with $r < k-1$, such that $L_U$ meets $\Omega$ in a canonical $\fq$-subgeometry of $\Omega$.
Then
\begin{equation}
  \label{eq:debeulebound}
  |L_U| \geq q^{n-1} + \ldots + q^{n-r} +  \frac{q^{k-r}-1}{q-1}.
 \end{equation}
If $r<k-1$, the equality holds if and only if $L_U$ is a canonical subgeometry $\PG(k-1,q)$ in $\PG(k-1,q^m)$. 
\end{theorem}

Moreover, the rank of a linear set is determined by its size and the minimum weight of its points, indeed the following holds.

\begin{proposition} \cite[Proposition 3.17]{SamPaolo}\label{pro:sizeminweight}
 Let $L_U$ be an $\fq$-linear set in $\PG(k-1,q^m)$, containing more than one point.
 Denote $e = \min_{P \in L_U} w_{L_U}(P)$.
 Then the rank of $L_U$ is the unique integer $n$ satisfying
 \[
  q^{n-e} < |L_U| \leq \frac{q^{n}-1}{q^e-1},
 \]
 i.e.\ $n = \lfloor \log_q(|L_U|) \rfloor + e$. Moreover, if $L_U$ spans the entire space then 
 \[
  q^{n-e} +\frac{q^{k-1}-1}{q-1} \leq  |L_U| \leq \frac{q^{n}-1}{q^e-1},
 \]
\end{proposition}

\subsection{Duality of linear sets and $\F_q$-subspaces}\label{sec:Dualityfq-sub}

Now, we recall the notion of the dual of a linear set.
Let $\sigma \colon V \times V \rightarrow \F_{q^m}$ be a nondegenerate reflexive sesquilinear form on the $k$-dimensional $\F_{q^m}$-vector space $V$ and consider \[
\begin{array}{cccc}
    \sigma': & V \times V & \longrightarrow & \F_q  \\
     & (x,y) & \longmapsto & \mathrm{Tr}_{q^m/q} (\sigma(x,y)).
\end{array}
\] 
So, $\sigma'$ is a nondegenerate reflexive sesquilinear form on $V$ seen as an $\fq$-vector space of dimension $km$. Then we may consider $\perp$ and $\perp'$ as the orthogonal complement maps defined by $\sigma$ and $\sigma'$, respectively, and $\tau$ and $\tau'$ as the polarities of $\PG(V,\F_{q^m})$ and $\PG(V,\fq)$ induced by $\sigma$ and $\sigma'$, respectively. For an $\F_q$-linear set in $\PG(V,\F_{q^m})$ of rank $n$, the $\F_q$-linear set $L_U^\tau=L_{U^{\perp'}}$ in $\PG(V,\F_{q^m})$ of rank $km-n$ is called the \textbf{dual linear set of $L_U$} with respect to the polarity $\tau$.
The definition of the dual linear set does not depend on the choice of the polarity.
Indeed, in \cite[Proposition 2.5]{polverino2010linear} is proved that, if $\tau_1$ and $\tau_2$ are two polarities and $\perp_1'$ and $\perp_2'$ are the orthogonal complement maps as above, then the dual linear sets $U^{\perp_1'}$ and $U^{\perp_2'}$ are $\mathrm{\Gamma L}(k,q^m)$-equivalent.

Moreover, we have the following relation between the weight of a subspace with respect to a linear set and the weight of its polar space with respect to the dual linear set. This property mainly relies on the fact that for any $\F_{q^m}$-subspace of $V$ holds $W^{\perp'}=W^{\perp}$.

\begin{proposition} \cite[Property 2.6]{polverino2010linear}\label{prop:weightdual}
Let $L_U \subseteq \PG(k-1,q^m)$ be an $\fq$-linear set of rank $n$. Let $\Omega_s=\PG(W,\F_{q^m})$ be an $s$-space of $\PG(k-1,q^m)$.
Then 
\[
w_{L_U^{\tau}}(\Omega_s^{\tau})=w_{L_U}(\Omega_s)+km-n-(s+1)m,
\]
i.e.
\[ \dim_{\fq}(U^{\perp'}\cap W^{\perp})=\dim_{\fq}(U\cap W)+\dim_{\fq}(V)-\dim_{\fq}(U)-\dim_{\fq}(W). \]
\end{proposition}

We can characterize maximum scattered linear sets via its dual.

\begin{proposition}\cite[Theorem 3.5]{polverino2010linear}\label{prop:scattdual}
Let $L_U \subseteq \PG(k-1,q^m)$ be an $\fq$-linear set of rank $mk/2$. Then $L_U$ is scattered if and only if $L_{U^{\perp'}}$ is scattered.
\end{proposition}

\subsection{Geometric dual of a rank metric code}

We recall an operation recently introduced on rank metric codes called  \emph{geometric dual}, which takes any element in $\mathfrak{C}[n,k,d]_{q^m/q}$ and gives another element in $\mathfrak{C}[mk-n,k,d']_{q^m/q}$. 

\begin{definition}
Let $\C$ be a non-degenerate $[n,k,d]_{q^m/q}$ and let $U$ be a system associated with $\C$.
Suppose also that $d_{k-1}^{\mathrm{rk}}(\C)\geq n-m+1$.
Then a \textbf{geometric dual} $\C^{\perp_{\mathcal{G}}}$ of $\C$ (with respect to $\perp'$) is defined as $\C'$, where $\C'$ is any code associated with the system $U^{\perp'}$, where $\perp'$ is defined as in Section \ref{sec:Dualityfq-sub}.
\end{definition}

This definition is justified by the following result.

\begin{theorem}\cite{BoreZull0}
Let $\C$ be an $[n,k,d]_{q^m/q}$ code, and let $U$ be a system associated with $\C$.
Suppose also that $d_{k-1}^{\mathrm{rk}}(\C)\geq n-m+1$. Then, up to equivalence, a geometric dual $\C^{\perp_{\mathcal{G}}}$ of $\C$ does not depend on the choice of the associated system and on the choice of code in $[\C]$, hence $\perp_{\mathcal{G}}$ is well-defined.
Moreover, $[\C^{\perp_{\mathcal{G}}}] \in \mathfrak{C}[km-n,k,d']_{q^m/q}$ for some $d'$.
\end{theorem}

\begin{remark}\label{rk:geogenrank}
Note that $U^{\perp'}$ is a system if and only if  $U^{\perp'}$ is not contained in any hyperplane of $\F_{q^m}^k$, which dually corresponds to $U$ not containing any $1$-dimensional $\F_{q^m}$-subspace of $\F_{q^m}^k$. By \eqref{eq:defgenrankweight}, this corresponds to require that $d_{k-1}^{\mathrm{rk}}(\C)\geq n-m+1$.
\end{remark}

This operation has been exploited  in \cite{BoreZull0} in the context of sum-rank metric codes.

\section{Bounds on the number of maximum weight codewords}\label{sec:bounds}

In this section we give upper and lower bounds on $M(\C)$, by using old and new bounds on linear sets.
In order to clarify the techniques and the arguments involved, we divide the analysis according to whether $\dim_{\F_{q^m}}(\C)$ is two or greater than two and $n\leq m$ or $n\geq m$.

\subsection{Dimension two case and $n\leq m$}

We start by describing the geometric meaning of $M(\C)$.

\begin{proposition}\label{prop:points-codewords2}
Let $\C$ be a non-degenerate $[n,2]_{q^m/q}$ code and let $U$ be any of its associated system.
Assume that $n\leq m$. Then 
\[M(\C)=(q^m-1) |\mathrm{PG}(1,q^m)\setminus L_U|.\]
\end{proposition}
\begin{proof}
Let $G$ be a generator matrix of $\C$ whose $\F_q$-span of its columns is $U$. Then by \eqref{eq:relweight}, for any nonzero $u \in \F_{q^m}^2$, we have
\[w_{L_U}(P)=n-w(uG),\]
where $P=\langle u^{\perp} \rangle_{\F_{q^m}} \in \PG(1,q^m)$
and hence the assertion follows from the fact that $P \notin L_U$ if and only if $w_{L_{U}}(P)=0$.
\end{proof}

Now, we are ready to give bounds on $M(\C)$.

\begin{theorem}\label{thm:boundsMC21}
Let $\C$ be a non-degenerate $[n,2]_{q^m/q}$ code and assume that $n\leq m$.
Then
\begin{equation}\label{eq:boundmC2}
q^{2m}-1-(q^m-1)  \frac{q^n-1}{q-1} \leq M(\C) \leq q^{2m}-1-(q^m-1)(q+1) .
\end{equation}
Moreover, if $\C$ contains a codeword of weight $n-1$,
\begin{equation}\label{eq:boundmC2scatt1}
M(\C) \leq q^{2m}-1-(q^m-1)(q^{n-1}+1),
\end{equation}
that is
\[
q^{2m}-q^{m+n-1}-\ldots -q^m +q^{n-1}+\ldots+q \leq M(\C) \leq q^{2m}-q^{m+n-1}-q^m+q^{n-1}.
\]
\end{theorem}
\begin{proof}
Let $U$ be any associated system with $\C$. Since $\langle U \rangle_{\F_{q^m}}=\mathbb{F}_{q^m}^2$ then $L_U$ cannot be a point and by \eqref{eq:cong1} we have
\[|L_U| \geq q+1.\]
Moreover, by \eqref{eq:card} 
\[|L_U| \leq \frac{q^n-1}{q-1}.\]
Therefore, by Proposition \ref{prop:points-codewords2} the first part of the assertion follows.
Now, let $G$ be any generator matrix of $\C$ and assume that $\C$ contains a codeword $uG$ of weight $n-1$. \eqref{eq:relweight} implies that there exists a point $P \in L_U$ with $w_{L_U}(P)=1$. By Theorem \ref{th:minsize}, we have
\[|L_U| \geq q^{n-1}+1\]
and hence again Proposition \ref{prop:points-codewords2} provides the upper bound on $M(\C)$.
\end{proof}

The above bounds \eqref{eq:boundmC2} can be improved once we know the second maximum weight of the code, extending the second part of the above theorem.

\begin{theorem}\label{thm:2dim2}
Let $\C$ be a non-degenerate $[n,2]_{q^m/q}$ code. Assume that $n\leq m$ and that $n-e$ is the second maximum weight. Then
\[q^{2m}-1-(q^m-1)\frac{q^n-1}{q^e-1} \leq M(\C) \leq q^{2m}-1-(q^m-1)(q^{n-e}+1),\]
i.e. $n =  \lfloor \log_q(q^m+1-\frac{M(\C)}{q^m-1}) \rfloor + e$. \\
Moreover, if $m=n$ then $e \mid m$.
\end{theorem}
\begin{proof}
Let $U$ be any associated system with $\C$. \eqref{eq:relweight} implies that $e = \min_{P \in L_U} w_{L_U}(P)$.
 Then by Proposition \ref{pro:sizeminweight}, we get
 \[
  q^{n-e} +1 \leq  |L_U| \leq \frac{q^{n}-1}{q^e-1}.
 \]
and by Proposition \ref{prop:points-codewords2} the bounds follows.
In particular, when $m=n$, then  
\cite[Proposition 3.1]{polverino2022divisible} implies that $e \mid m$ and $U$ is an $\F_{q^e}$-subspace of $\mathbb{F}_{q^m}^2$. 
\end{proof}

\begin{remark}
In the case in which $m=n$ and $n-e$ is the second maximum weight, the code $\C$ turns out to be $e$-divisible, that is all the weights of the codewords are multiply of $e$; see \cite{polverino2022divisible}.
\end{remark}

\subsection{Dimension two case and $n> m$}

We will now deal with the case $n>m$. To this aim we will need the aid of the dual of linear sets.

\begin{theorem}\label{thm:boundMC22}
Let $\C$ be a non-degenerate $[n,2]_{q^m/q}$ code and assume that $m< n$ and $d(\C)\geq n-m+1$.
Then
\begin{equation}\label{eq:boundmC2dual}
q^{2m}-1-(q^m-1)  \frac{q^{2m-n}-1}{q-1} \leq M(\C) \leq q^{2m}-1-(q^m-1)(q+1) .
\end{equation}

If $m-e$ is the second maximum weight of $\C$, then
\[q^{2m}-1-(q^m-1)\frac{q^{2m-n}-1}{q^e-1} \leq M(\C) \leq q^{2m}-1-(q^m-1)(q^{2m-n-e}+1),\]
and $n =2m- \lfloor \log_q(q^m+1-\frac{M(\C)}{q^m-1}) \rfloor - e$.
\end{theorem}
\begin{proof}
Let $U$ be any system associated with $\C$.
As for Proposition \ref{prop:points-codewords2}, we may determine $M(\C)$ by determining the number of points of weight $n-m$ in $L_U$. 
Consider $L_{U^{\perp'}}$ the dual linear set of $L_U$. Proposition \ref{prop:weightdual} implies that a point $P$ is such that $w_{L_U}(P)=n-m$ if and only if $w_{L_{U^{\perp'}}}(P^\tau)=0$. Hence, 
\begin{equation}\label{eq:M(C)dual}
M(\C)=(q^m-1) |\mathrm{PG}(1,q^m)\setminus L_{U^{\perp'}}|.
\end{equation}
Since the rank of $L_{U^{\perp'}}$ is $2m-n < m$, by \eqref{eq:card} we have
\[ q+1\leq |L_{U^{\perp'}}| \leq \frac{q^{2m-n}-1}{q-1}, \]
since $|L_{U^{\perp'}}|>1$ by Remark \ref{rk:geogenrank}.
If $\C$ contains a codeword of weight $m-e$, then there exists a point $P$ such that $w_{L_U}(P)=n-m+e$ and hence $w_{L_{U^{\perp'}}}(P^\tau)=e$, by Proposition \ref{prop:weightdual}. Now, by applying Proposition \ref{pro:sizeminweight}
\[ |L_{U^{\perp'}}|\geq q^{2m-n-e}+1.\]
The bounds follow by \eqref{eq:M(C)dual}.
\end{proof}

\subsection{Larger dimension case and $n\leq m$}

In this section we assume that $n\leq m$ and $k>2$.
In order to underline the second order of magnitude in $M(\C)$, we will write the bounds not directly on $M(\C)$ but on $\frac{M(\C)}{q^m-1}$.
Under these assumptions, $M(\C)$ corresponds to determine the number of external hyperplanes to a linear set.

\begin{proposition}\label{prop:points-codewordsk}
Let $\C$ be a non-degenerate $[n,k]_{q^m/q}$ code and let $U$ be any associated system. Assume that $n\leq m$. Then 
\[M(\C)=(q^m-1) |\{ H=\PG(k-2,q^m) \colon H \cap L_U=\emptyset \}|.\]
\end{proposition}
\begin{proof}
Let $G$ be a generator matrix of $\C$ such that the $\fq$-span of its columns is $U$. Then by \eqref{eq:relweight}, a codeword $c=uG$ has maximum weight if and only if
\[w_{L_U}(u^\perp)=n-w(uG)=0,\]
and hence the assertion follows.
\end{proof}

To prove our bounds, we need the following two geometric lemmas.

\begin{lemma}\label{lem:hypersubg}
Let $\Sigma=\PG(k-1,q)$ be a canonical subgeometry in $\PG(k-1,q^m)$ and $k\leq m$. Then the number of hyperplanes of $\PG(k-1,q^m)$ meeting $\Sigma$ in at least one point is
\[\alpha=\frac{q^{mk}-1}{q^m-1}-\prod_{i=1}^{k-1}(q^m-q^i),\]
and the number of hyperplanes of $\PG(k,q^m)$ meeting $\Sigma$ is at least two points (and hence at least $q+1$) is
\[\beta=\frac{q^{mk}-1}{q^m-1}-\prod_{i=1}^{k-1}(q^m-q^i)-\frac{q^k-1}{q-1}
\prod_{i=0}^{k-2}(q^m-q^i) .\]
\end{lemma}
\begin{proof}
Let $U$ be any $k$-dimensional $\fq$-subspace of $\F_{q^m}^k$ such that $L_U=\Sigma$. Clearly, $U$ is a $[k,k,1]_{q^m/q}$ system. Then any code $\C$ associated with $U$ coincide with $\F_{q^m}^k$ (and it is trivially an MRD code). Moreover, by \eqref{eq:relweight} the number $\gamma$ of external hyperplanes corresponds to the  number of codewords of $\C$ with weight $k$ divided by $q^m-1$, more precisely
\[\gamma = \frac{A_k}{q^m-1},\]
where $A_k$ denotes the number of vectors in $\F_{q^m}^k$ with weight $k$.
It is an easy computation to see that 
\[A_k=\prod_{i=0}^{k-1}(q^m-q^i).\]
Therefore,
\[\gamma=\prod_{i=1}^{k-1}(q^m-q^i)\]
and hence
\[\alpha=\frac{q^{mk}-1}{q^m-1}-\gamma.\]
The value of $\beta$ can be also obtained by subtracting to the number of hyperplanes in $\PG(k-1,q^m)$ the value of $\gamma$ and the number $\delta$ of the tangent hyperplanes to $\Sigma$.
As before, one can see that 
\[\delta=\frac{A_{k-1}}{q^m-1}, \]
where $A_{k-1}$ denotes the number of vectors in $\F_{q^m}^k$ with weight $k-1$. By Theorem \ref{th:weightdistribution} and \cite{migler2004weight}, it follows that
\[ A_{k-1}={k\brack 1}_q \sum_{t=0}^{k-2} (-1)^{t-k}{k-1 \brack k-2-t}_q q^{\binom{k-2-t}{2}}(q^{m(t+1)}-1)=\frac{q^k-1}{q-1}
\prod_{i=0}^{k-2}(q^m-q^i).
\]
\end{proof}

\begin{lemma} \label{lem:tangenthyperlinset}
Let $L_U$ be an $\F_q$-linear set spanning $\PG(k-1,q^m)$ having rank $n$ with $3 \leq k \leq n \leq m$. Then for each point $P \in L_U$, there exists an $r$-space through $P$ meeting $L_U$ exactly in $P$, for each $r \in \{0,\ldots,k-2\}$. In particular, there exists a hyperplane through $P$ which is tangent to $L_U$.
\end{lemma}
\begin{proof}
Suppose by contradiction that all the lines through $P$ are not tangent in $P$ and so they meet $L_U$ in at least $q+1$ points.  Then 
\[
    q \frac{q^{m(k-1)}-1}{q^m-1}+1 \leq \lvert L_U \rvert \leq \frac{q^n-1}{q-1},
\]
and so we obtain a contradiction since $n \leq m$.
Now, suppose the assertion holds for any $r \in \{1,\ldots,t\}$ and let $\pi$ be a $t$-space through $P$ which is tangent to $L_U$. Suppose by contradiction that all the $q^m+1$ $(t+1)$-spaces through $\pi$ are not tangent $L_U$ then
\[
    q (q^m+1)+1 \leq \lvert L_U \rvert \leq \frac{q^n-1}{q-1},
\]
again a contradition to the fact that $n\leq m$.
\end{proof}

We can use Lemma \ref{lem:hypersubg} to obtain upper and lower bounds, using the fact that in a system $U$ in $\F_{q^m}^k$ we always find $k$ $\F_{q^m}$-linearly independent vectors, which geometrically means that in $L_U$ is contained a canonical subgeometry.

\begin{theorem}\label{thm:boundn<=m}
Let $\C$ be a non-degenerate $[n,k]_{q^m/q}$ code. Assume that $n\leq m$. Then 
\begin{equation} \label{eq:boundM(C)n<m}
\frac{q^{mk}-1}{q^m-1}-\frac{q^n-1}{q-1}\frac{q^{(k-1)m}-1}{q^m-1} + q\beta \leq \frac{M(\C)}{q^m-1}\leq  \prod_{i=1}^{n-1}(q^m-q^i) 
\end{equation}
where
\[\beta=\frac{q^{mk}-1}{q^m-1}-\prod_{i=1}^{k-1}(q^m-q^i)-\frac{q^k-1}{q-1}
\prod_{i=1}^{k-2}(q^m-q^i) .\]
\end{theorem}
\begin{proof}
Let $U$ be any associated system with $\C$. 
By Proposition \ref{prop:points-codewordsk}, we need to count the number of external hyperplanes to $L_U$ in $\PG(k-1,q^m)$.
Denote by $\tau_0,\tau_1$ and $\tau_s$ the number of hyperplanes in $\PG(k-1,q^m)$ meeting $L_U$ in $0,1$ and at least two points, respectively.
Clearly, $\tau_0+\tau_1+\tau_s=\frac{q^{mk}-1}{q^m-1}$.
Note that, by \eqref{eq:cong1} if a hyperplane meets $L_U$ is at least two points, then the intersection will contain at least $q+1$ points. Therefore, by double counting the set
\[\{ (P,H) \colon P \text{ point}, H \text{ hyperplane} \text{ and } P\in L_U \cap H \},\]
and using that any secant hyperplane meets $L_U$ in at least $q+1$ points, we obtain
\[ \tau_1+\tau_s (q+1) \leq |L_U| \frac{q^{(k-1)m}-1}{q^m-1}, \]
from which we derive
\[ \tau_1+\tau_s\leq |L_U| \frac{q^{(k-1)m}-1}{q^m-1} -q\tau_s. \]
Note that $U$ contains $k$ vectors $u_1,\ldots,u_k$ which are $\F_{q^m}$-linearly independent.
Denote by $W=\langle u_1,\ldots,u_k \rangle_{\fq}$, then $L_W=\PG(k-1,q)$ is a canonical subgeometry contained in $L_U$. Therefore,
\[ \tau_s \geq \beta,\]
where $\beta$ is the number of secant hyperplanes to $L_W$ (computed in Lemma \ref{lem:hypersubg}).
Therefore, we have 
\begin{equation}\label{eq:boundtactical}
\tau_1+\tau_s\leq |L_U| \frac{q^{(k-1)m}-1}{q^m-1} -q\tau_s\leq \frac{q^n-1}{q-1}\frac{q^{(k-1)m}-1}{q^m-1} -q\beta, 
\end{equation}
and hence
\[\tau_0 \geq \frac{q^{mk}-1}{q^m-1} -\frac{q^n-1}{q-1}\frac{q^{(k-1)m}-1}{q^m-1} +q\beta. \]
Moreover, we can upper bound $M(\C)$ with the number of matrices in $\fq^{m\times n}$ of rank $n$ which is
$(q^m-1)\prod_{i=1}^{n-1}(q^m-q^i)$.
\end{proof}

\begin{remark}\label{rk:exthyperssubgeo}
Note that when $n=k$
\[ (q^m-1)\prod_{i=1}^{k-1}(q^m-q^i)= q^{mk}-1-(q^m-1)\alpha,\]
where $\alpha$ is the number of hyperplanes of $\PG(k-1,q^m)$ meeting $\PG(k-1,q)$ in at least one point.
\end{remark}

We can prove another upper bound on $M(\C)$ in which, unlike the previous bound, also the length of the code is involved.

\begin{theorem}\label{thm:boundn<=m2}
Let $\C$ be a non-degenerate $[n,k]_{q^m/q}$ code. Assume that $n\leq m$ and $n-e$ is the second maximum weight of $\C$. Then 
\[ q^{m(k-1)}-q^{m(k-2)+n-e}-q^{m(k-2)} \left(\frac{q^{n-e}-1}{q^e-1}\right) \leq \frac{M(\C)}{q^m-1}\leq q^{m(k-1)}-q^{m(k-2)+n-e},\]
i.e., 
$m(k-2)+n  = \lfloor \log_q(q^{m(k-1)}-\frac{M(\C)}{q^m-1}) \rfloor + e$.
\end{theorem}
\begin{proof}
Let $U$ be any associated system with $\C$ and let $c \in \C$ with $w(c)=n-e$.
We determine a lower and an upper bound on the number of external hyperplanes to $L_U$ in $\PG(k-1,q^m)$.
Since $n-e$ is the second maximum weight in $\C$, by Theorem \ref{th:connection} we have that
\begin{equation} \label{eq:minweighthyperplanes}
e=\min \{\dim_{\fq}(U\cap H) \colon  H\mbox{ is an } \F_{q^m}\mbox{-hyperplane of }\F_{q^m}^k \mbox{ such that }U \cap H \neq \{0\}\}.
\end{equation}
We prove that, $e=\min \{w_{L_U}(P) \colon  P \in L_U\}$. Suppose that there exists a point $Q \in L_U$ such that $w_{L_U}(Q)=e'<e$. Then by Lemma \ref{lem:tangenthyperlinset}, there exists a hyperplane $H'$ through $Q$ that is tangent to $L_U$. This means that $\dim_{\F_q}(U \cap H')=e'$, contradicting \eqref{eq:minweighthyperplanes}. On the other hand, if $w_{L_U}(P) \geq e+1$, for every $P \in L_U$, then we have $\dim_{\fq}(U\cap H) \geq e+1$, for every  $H$ hyperplane of $\F_{q^m}^k \mbox{ such that }U \cap H \neq \{0\}$, contradicting again \eqref{eq:minweighthyperplanes}. \\
Since $w(c)=n-e$, then by Theorem \ref{th:connection} there exists a projective hyperplane $\pi=\PG(H,\F_{q^m})$ such that $w_{L_U}(\pi)=e$.
Since $w_{L_U}(Q)\geq e$ for any point $Q\in L_U$, then $\pi \cap L_U=\{P\}$, for some point $P$.
Let $\PG(k-3,q^m)=\PG(W,\F_{q^m})$ be any projective hyperplane of $\pi$ not containing $P$, i.e. $W\cap U=\{0\}$, and let $\overline{U}=(U+W)/W$, which is an $\fq$-subspace of the quotient $V=\F_{q^m}^k/W=V(2,q^m)$.
Since $U\cap W=\{0\}$, $L_{\overline{U}}$ is an $\fq$-linear set of rank $n$ contained in $\PG(1,q^m)=\PG(V,\F_{q^m})$ having $\langle (U\cap H + W)/W \rangle_{\F_{q^m}}$ as a point of weight $e$, and all the other points of weight greater than or equal than $e$. Hence, by Proposition \ref{pro:sizeminweight}
\[ q^{n-e}+1 \leq |L_{\overline{U}}|\leq \frac{q^n-1}{q^e-1}. \]
Therefore, 
\begin{equation}\label{eq:boundquotient}
|L_{\overline{U}}|=q^{n-e}+c+1,
\end{equation}
for some integer $c$ such that 
\begin{equation}\label{eq:19}
0\leq c \leq \frac{q^{n-e}-1}{q^e-1}-1.
\end{equation}
Now, the size of $L_{\overline{U}}$ is the number of projective hyperplanes $\PG(W,\F_{q^m})$ meeting $L_U$ in at least one point.
The number of projective hyperplanes in $\PG(k-1,q^m)$ passing through $P$ is 
\[\frac{q^{m(k-1)}-1}{q^m-1},\]
whereas the number of projective hyperplanes in $\pi$ not passing through $P$ is
\[\frac{q^{m(k-1)}-1}{q^m-1}-\frac{q^{m(k-2)}-1}{q^m-1}=q^{m(k-2)}.\]
Denote by $\pi'_1=\PG(W_1,\F_{q^m}),\ldots,\pi'_{q^{m(k-2)}}=\PG(W_{q^{m(k-2)}},\F_{q^m})$ the projective hyperplanes of $\pi$ not passing through $P$ and, because of \eqref{eq:boundquotient} we can write
\[ |L_{(U+W_i)/W_i}|=q^{n-e}+c_i+1, \]
for any $i$. Plugging together all of the above information, the number of projective hyperplanes meeting $L_U$ in at least one point is
\[ \frac{q^{m(k-1)}-1}{q^m-1}+\sum_{i=1}^{q^{m(k-2)}} (q^{n-e}+c_i) \]
that is
\[ q^{m(k-2)+n-e}+\sum_{i=1}^{q^{m(k-2)}} c_i + \frac{q^{m(k-1)}-1}{q^m-1}.  \]
Therefore, the number of external hyperplanes to $L_U$ is
\[ \frac{M(\C)}{q^m-1}=\frac{q^{mk}-1}{q^m-1}-q^{m(k-2)+n-e}-\sum_{i=1}^{q^{m(k-2)}} c_i - \frac{q^{m(k-1)}-1}{q^m-1}, \]
i.e.,
\[ \frac{M(\C)}{q^m-1}=q^{m(k-1)}-q^{m(k-2)+n-e}-\sum_{i=1}^{q^{m(k-2)}} c_i , \]
and hence the assertion follows by \eqref{eq:19}.
\end{proof}

\begin{remark}\label{rk:disjoint}
The bounds of the above theorem depends on the second maximum weight and the possible values of $M(\C)$ are in disjoint intervals (according to $e$).
Moreover, once four values among $m,n,q, M(\C)$ and the second maximum weight $m-e$ are known, then one can determine the remaining one directly from the relation $m(k-2)+n  = \lfloor \log_q(q^{m(k-1)}-\frac{M(\C)}{q^m-1}) \rfloor + e$.
\end{remark}

\subsection{Larger dimension case and $m\leq n$}

In this section we will now deal with the case in which $n\geq m$. As for the previous section, we give a geometric interpretation for the value of $M(\C)$, which corresponds to count the number of external points to the dual of the linear set defined by a system associated with $\C$.

\begin{proposition}\label{prop:points-codewordsk2}
Let $\C$ be a non-degenerate $[n,k]_{q^m/q}$ code and let $U$ be any associated system. Assume that $m\leq n$. Then 
\[M(\C)=(q^m-1) |\{ H=\PG(k-2,q^m) \colon w_{L_U}(H)=n-m \}|\]
\[ =(q^m-1) |\PG(k-1,q^m)\setminus L_{U^{\perp'}}|. \]
\end{proposition}
\begin{proof}
Let $G$ be a generator matrix of $\C$ such that the $\fq$-span of its columns is $U$. Then by \eqref{eq:distancedesign}, a codeword $c=uG$ has maximum weight if and only if
\[w_{L_U}(u^\perp)=n-w(uG)=n-m,\]
and hence the first equality follows. The second one follows by appyling Proposition \ref{prop:weightdual}.
\end{proof}

We can now derive bounds on $M(\C)$ by making use of the bounds on the number of points of linear sets.

\begin{theorem}\label{th:k>2weak}
Let $\mathcal{C}\subseteq \mathbb{F}_{q^m}^n$ be a non-degenerate $[n,k]_{q^m/q}$ code. Assume that $m\leq n$ and $d_{k-1}^{\mathrm{rk}}(\C)\geq n-m+1$.
Then
\begin{equation}\label{eq:boundmC2dualk}
\frac{q^{km}-1}{q^m-1} -\frac{q^{km-n}-1}{q-1} \leq  \frac{M(\C)}{q^m-1} \leq \frac{q^{km}-1}{q^m-1} -\frac{q^k-1}{q-1}.
\end{equation}
In particular, if the second maximum weight of $\C$ is $m-e$,
\begin{equation}\label{eq:boundmC2scattk2}
 \frac{q^{km}-1}{q^m-1}-\frac{q^{km-n}-1}{q^e-1} \leq \frac{M(\C)}{q^m-1} \leq \frac{q^{km}-1}{q^m-1}-\left( q^{km-n-e}+\frac{q^{k-1}-1}{q-1}\right),
\end{equation}
i.e.
$km-n = \lfloor \log_q(\frac{q^{mk}-1}{q^m-1}-\frac{M(\C)}{q^m-1}) \rfloor + e$.
\end{theorem}
\begin{proof}
Let $U$ be any system associated with $\C$.
By Proposition \ref{prop:points-codewordsk2}, 
\[M(\C)=(q^m-1)|\PG(k-1,q^m)\setminus L_{U^{\perp'}}|.\]
Since $L_{U^{\perp'}}$ has rank $km-n\leq (k-1)m$ and $\langle L_{U^{\perp'}}\rangle =\PG(k-1,q^m)$ by Remark \ref{rk:geogenrank}, then
\[ \frac{q^k-1}{q-1}\leq |L_{U^{\perp'}}|\leq \frac{q^{km-n}-1}{q-1}.\]
Moreover, if the second maximum weight in $\C$ is $m-e$ then this means that there exists a point $P$ such that $w_{L_{U^{\perp'}}}(P)=e$ and $w_{L_{U^{\perp'}}}(Q)\geq e$ for any point $Q\in L_{U^{\perp'}}$, because of \eqref{eq:relweight} and Proposition \ref{prop:weightdual}.
Therefore, by Proposition \ref{pro:sizeminweight} we have
\[ q^{km-n-e}+\frac{q^{k-1}-1}{q-1} \leq  |L_{U^{\perp'}}|\leq \frac{q^{km-n}-1}{q^e-1} \]
and \eqref{eq:boundmC2scattk2} follows.
\end{proof}

\begin{remark}\label{rk:disjoint2}
The same remark as Remark \ref{rk:disjoint} applies to the bounds \eqref{eq:boundmC2scattk2}.
\end{remark}

The above lower bound \eqref{eq:boundmC2scattk2} can be proved with less restrictive hypothesis but with a more involved condition.

\begin{theorem}\label{th:boundsubgeometrycodes}
Let $\C$ be a non-degenerate $[n,k]_{q^m/q}$ code and assume that $m\leq n$ and $d_{k-1}^{\mathrm{rk}}(\C)\geq n-m+1$. 
Let $G'$ be any of generator matrix of $\C^{\perp_{\mathcal{G}}}$. Suppose there exist $r\geq 1$ codewords $c_1,\ldots,c_r \in \C^{\perp_{\mathcal{G}}}$ $\F_{q^m}$-linearly independent such that the $\fq$-subspace
\[ W=\psi_{G'} \left( \bigcap_{i=1}^r \mathrm{supp}(c_i)^\perp \right)\]
satisfies $\dim_{\fq}(W)=\dim_{\F_{q^m}}(\langle W \rangle_{\F_{q^m}})=k-r$.
Then
\begin{equation}\label{eq:boundmC2scattk2boundnew}
\frac{M(\C)}{q^m-1} \leq \frac{q^{km}-1}{q^m-1}-\left( q^{km-n-1}+\ldots+q^{km-n-k+r}+\frac{q^r-1}{q-1}\right).
\end{equation}
\end{theorem}
\begin{proof}
Let $U$ be the system associated with $\C$ such that $U^{\perp'}$ is the $\fq$-span of the columns of $G'$. Note that $U^{\perp'}$ has dimension $km-n$.
Since $c_i \in \C^{\perp_{\mathcal{G}}}$, for any $i \in \{1,\ldots,r\}$ there exists $v_i \in \F_{q^m}^k$ such that $c_i=v_i G'$ and $v_1,\ldots,v_r$ are $\F_{q^m}$-linearly independent. Therefore, by Theorem \ref{th:suppogeom}
\[ W=\psi_{G'} \left( \bigcap_{i=1}^r \mathrm{supp}(c_i)^\perp \right)= (U^{\perp'}\cap v_1^{\perp})\cap \ldots (U^{\perp'}\cap v_r^{\perp})=U^{\perp'}\cap (v_1^{\perp}\cap\ldots\cap v_r^{\perp}), \]
and note that $v_1^{\perp}\cap\ldots\cap v_r^{\perp}=\langle v_1,\ldots,v_r\rangle_{\F_{q^m}}^{\perp}$ is an $\F_{q^m}$-subspace of $\F_{q^m}^k$ having dimension $k-r$ meeting $U^{\perp'}$ in $W$, which is an $\fq$-subspace such that $\dim_{\fq}(W)=\dim_{\F_{q^m}}(\langle W \rangle_{\F_{q^m}})=k-r$, contained in $L_{U^{\perp'}}$. Therefore, $L_W$ is a canonical subgeometry $\PG(k-r-1,q)$ in $\PG(v_1^{\perp}\cap\ldots\cap v_r^{\perp},\F_{q^m})=\PG(k-r-1,q^m)$ contained in $L_{U^{\perp'}}$. 
By Theorem \ref{thm:newBound} we have
\[ |L_{U^{\perp'}}| \geq q^{km-n-1}+\ldots+q^{km-n-k+r}+\frac{q^r-1}{q-1}. \]
The bound then follows from Proposition \ref{prop:points-codewordsk2}.
\end{proof}

\begin{remark}
Note that Theorem \ref{th:boundsubgeometrycodes} extends \eqref{eq:boundmC2scattk2} of Theorem \ref{th:k>2weak} since the property of containg a codeword of weight $m-1$ is equivalent to require the existence of $k-1$ $\F_{q^m}$-linearly independent codewords in $\C^{\perp_{\mathcal{G}}}$ such that 
\[ \dim_{\fq}\left(\psi_{G'} \left( \bigcap_{i=1}^r \mathrm{supp}(c_i)^\perp \right)\right)=1, \]
and hence clearly satisfies the assumption of the aforementioned theorem.
\end{remark}

\begin{remark}
Following the proof of the above result, the assumption $\dim_{\fq}(W)=$ \newline $\dim_{\F_{q^m}}(\langle W \rangle_{\F_{q^m}})=k-r$ is equivalent to the existence of a projective subspace $\Omega$ of codimension $r$ meeting $L_{U^{\perp'}}$ is a canonical subgeometry of $\Omega$.
Indeed, this allowed us to use Theorem \ref{thm:newBound}. 
\end{remark}

\section{Equality in the lower bounds}\label{sec:lowerbounds}

In this section we study the case of equality in the lower bounds determined in the previous section.
We start with a geometric characterization of $\mathbb{F}_{q^m}$-linear MRD codes in $\mathbb{F}_{q^m}^n$ with dimension two, as an easy consequence of the geometric correspondence described in Section \ref{sec:georankmetric}.

\begin{proposition}\label{prop:2dimscatt}
Let $\C$ be a non-degenerate $[n,2,d]_{q^m/q}$ code and let $U$ be any associated system.
If $\C$ is not the simplex code, then $\C$ is MRD if and only if $U$ is a scattered subspace and $n\leq m$.
\end{proposition}
\begin{proof}
Since $\dim_{\mathbb{F}_{q^m}}(\C)=2$, then by Theorem \ref{th:singletonrank} the code $\C$ is MRD if and only if $2m=\max\{m,n\}(\min\{m,n\}-d+1)$.
If $n\leq m$, then $d=n-1$ and by \eqref{eq:distancedesign} this is equivalent to require that $U$ is a scattered $\F_q$-subspace.
If $m<n$, then by \eqref{eq:distancedesign} $n\mid 2m$ and hence $n=2m$, which implies that $U=\mathbb{F}_{q^m}^2$. 
\end{proof}

We are now ready to characterize the rank metric codes of dimension two satisfying the lower bounds on $M(\C)$.

\begin{theorem}\label{thm:charmin2}
Let $\mathcal{C}\subseteq \mathbb{F}_{q^m}^n$ be a non-degenerate $[n,2,d]_{q^m/q}$ code and assume $d \geq n-m+1$. Then $M(\C)$ is minimum if and only if either $\C$ or $\C^{\perp_{\mathcal{G}}}$ is an MRD code.
\end{theorem}
\begin{proof}
Let $U$ be any system associated with $\C$.
Assume first that $n\leq m$. $M(\C)$ satisfies the equality in the lower bound of Theorem \ref{thm:boundsMC21} if and only if 
\[M(\C)=q^{2m}-1-(q^m-1)\frac{q^n-1}{q-1},\]
i.e. $|L_U|=(q^n-1)/(q-1)$, that is if and only if $U$ is a scattered $\fq$-subspace of $\F_{q^m}^2$.
By Proposition \ref{prop:2dimscatt}, this implies that $\C$ is an MRD code.
Suppose now that $n\geq m$. $M(\C)$ satisfies the equality in the lower bound of Theorem \ref{thm:boundMC22} if and only if
\[M(\C)=q^{2m}-1-(q^m-1)  \frac{q^{2m-n}-1}{q-1},\]
which by \eqref{eq:M(C)dual} is equivalent to say that
\[|L_{U^{\perp'}}|=\frac{q^{2m-n}-1}{q-1}.\]
Since the rank of $L_{U^{\perp'}}$ is $2m-n$, then $L_{U^{\perp'}}$ is scattered and $\langle L_{U^{\perp'}}\rangle=\PG(1,q^m)$, so any code associated with $U^{\perp'}$ is an MRD code.
\end{proof}

When the dimension of the code is larger, then the variety of rank metric codes having the minimum number of codewords of maximum weight is much larger and then the family of MRD codes.

\begin{theorem}\label{thm:chargeominimumk}
Let $\C$ be a non-degenerate $[n,k]_{q^m/q}$ code and let $U$ be any associated system.
Assume that $m\leq n$. Then $M(\C)$ is minimum with respect to \eqref{eq:boundmC2dualk} if and only if $U^{\perp'}$ is a scattered $\fq$-subspace of $\F_{q^m}^k$.
In particular, if $M(\C)$ is minimum then $n\geq km/2$.
\end{theorem}
\begin{proof}
In this case, by Theorem \ref{th:k>2weak},
\[M(\C)=q^{km}-1-(q^m-1)\frac{q^{km-n}-1}{q-1},\]
and by Proposition \ref{prop:points-codewordsk2}, we have that
\[|L_{U^{\perp'}}|=\frac{q^{km-n}-1}{q-1}.\]
Since the rank of $L_{U^{\perp'}}$ is $km-n$, the above equality implies that $L_{U^{\perp'}}$ is scattered.
The last part follows by applying Theorem \ref{th:boundscattrank} to $U^{\perp'}$.
\end{proof}

As seen before, for the case of rank metric codes of dimension two, the MRD codes reach the equality in the lower bound on the number of codewords of maximum weight. 
The number of codewords of maximum weight in an $\F_{q^m}$-linear MRD code $\C$ in $\F_{q^m}^n$ of dimension $k$ with $m\leq n$ (note that its minimum distance is $d=m-km/n+1$) is given in Theorem \ref{th:weightdistribution} and it is
\[M(\C)= \sum_{t=0}^{\frac{km}n-1} (-1)^{t-\frac{km}n+1}{m \brack \frac{km}n-1-t}_q q^{\binom{\frac{km}n-1-t}{2}}(q^{n(t+1)}-1).\]

In the next result we prove that this value is the minimum if and only if $n=mk/2$.

\begin{proposition}
Let $\C$ be a non-degenerate $[n,k]_{q^m/q}$ MRD code and assume that $m\leq n$. Then $M(\C)$ is minimum with respect to \eqref{eq:boundmC2dualk} if and only if $n=mk/2$.
\end{proposition}
\begin{proof}
Let $U$ be any system associated with $\mathcal{C}$.
For an MRD code with these parameters, its minimum distance is $d=m-km/n+1$. By \eqref{eq:distancedesign}, it follows that 
\[ w_{L_U}(H)\leq n-(m-km/n+1), \]
for any hyperplane $H$ of $\PG(k-1,q^m)$ and there exists at least one hyperplane satisfying the equality.
By Proposition \ref{prop:weightdual},
\[ w_{L_{U^{\perp'}}}(P)\leq km-n-m(k-1)+n-(m-km/n+1)=\frac{km}n-1, \]
for any point $P \in \PG(k-1,q^m)$.
By Theorem \ref{thm:chargeominimumk}, $M(\C)$ is minimum if and only if $L_{U^{\perp'}}$ is scattered, which happens if and only if $\frac{km}n-1=1$, that is $n=km/2$.
\end{proof}

If $n=mk/2$, all the codes having the minimum number of codewords with maximum weight are MRD codes.

\begin{proposition}\label{prop:charMRDmk2}
Let $\C$ be a non-degenerate $[n,k]_{q^m/q}$ code and assume that $n=mk/2$. Then $M(\C)$ is minimum with respect to \eqref{eq:boundmC2dualk} if and only if $\C$ is an MRD code.
\end{proposition}
\begin{proof}
Let $U$ be any system associated with $\C$.
By Theorem \ref{thm:chargeominimumk}, $M(\C)$ is minimum if and only if $L_{U^{\perp'}}$ is scattered.
Since the rank of $L_{U^{\perp'}}$ is $km/2$ then by Proposition \ref{prop:scattdual} the linear set $L_U$ is scattered as well. 
By \cite[Theorem 3.2]{csajbok2017maximum}, the code $\C$ is an MRD code.
The converse follows by the above proposition.
\end{proof}

\begin{remark}
The MRD codes with parameters as in Proposition \ref{prop:charMRDmk2} are in correspondence with maximum scattered $\fq$-linear sets ($\fq$-subspaces) in $\PG(k-1,q^m)$ (in $\mathbb{F}_{q^m}^k$), see \cite[Theorem 4.8]{polverino2020connections} and also \cite{marino2022evasive,zini2021scattered}.
Such subspaces exist for any value of $q,m,k$ when $mk$ is even, as proved in a series of papers \cite{ball2000linear,bartoli2018maximum,blokhuis2000scattered,csajbok2017maximum}.
\end{remark}

As we will discuss in the next remark, not all of the scattered spaces give rise to an MRD code, therefore there is no hope to extend Theorem \ref{thm:charmin2} to larger dimension.

\begin{remark}
Consider $W=\{(x,x^q,a)\colon x \in \F_{q^m}, a \in \fq\} \subseteq \F_{q^m}^3$, with $m\geq 2$.
It results to be a scattered $\fq$-subspace of dimension $m+1$ (and $L_W$ defines a R\'edei type blocking set in $\PG(2,q^m)$). Then choose $U=W^{\perp'}$, so that $\dim_{\fq}(U)=2m-1$. Then $\langle U \rangle_{\F_{q^m}}=\F_{q^m}^3$, otherwise by Proposition \ref{prop:weightdual} there would exist a point $P\in \PG(2,q^m)$ such that $w_{L_W}(P)=m$, a contradiction to the fact that $W$ is scattered.
We can then consider a $[2m-1,3]_{q^m/q}$ rank metric code $\C$ associated with $U$ and $M(\C)$ is the minimum as $U^{\perp'}=W$ is scattered (cf. Theorem \ref{thm:chargeominimumk}).
The minimum distance of $\C$ is
\[d=2m-1-\max\{w_{L_U}(H) \colon H=\PG(1,q^m)\}\]
\[=m-\max\{w_{L_W}(P) \colon P\in\PG(2,q^m)\}=m-1,\]
by using again Proposition \ref{prop:weightdual}.
It is easy to see that the code $\C$ is not an MRD since its parameters do not attain the equality in \eqref{eq:boundgen}.
\end{remark}

\begin{remark}
In Theorem \ref{thm:2dim2} we also present another lower bound depending on the second minimum weight of a $2$-dimensional $\F_{q^m}$-linear non-degenerate rank metric code. Examples of codes attaining the equality in such a bound are the codes associated with the dual of a scattered $\mathbb{F}_{q^e}$-linear sets $L_{U}$ in $\PG(1,q^m)$, when $e \mid m$. Unfortunately, we do not know if this is the only case.
\end{remark}

We analyze the equality in the lower bound in Theorem \ref{thm:boundn<=m}. First, we prove a property on the points external to a subgeometry.

\begin{lemma} \label{lem:tangenthypersubgeo}
Let $\Sigma=\PG(k-1,q)$ be a canonical subgeometry $\PG(k-1,q)$ in $\PG(k-1,q^m)$ and assume that $3 \leq k \leq m$. Then for each point $P \in \PG(k-1,q^m) \setminus \Sigma$, there exists an $r$-space through $P$ meeting $\Sigma$ in exactly one point, for each $r \in \{0,\ldots,k-2\}$. In particular, there exists a hyperplane through $P$ which is tangent to $\Sigma$.
\end{lemma}
\begin{proof}
Let $P=\langle v\rangle_{\F_{q^m}} \in \PG(k-1,q^m)=\PG(V,\F_{q^m}) \setminus \Sigma$. We will prove it by induction on $r$. Let $\Sigma=L_W$, where $W$ is a $k$-dimensional $\F_q$-subspace of $V$ such that $\langle W \rangle_{\F_{q^m}}=V$. Suppose by contradiction that all the lines through $P$ are external or meet $\Sigma$ in at least $q+1$ points. Let consider the $\F_q$-linear set $ L_{W'} $ in $\PG(V/\langle v\rangle_{\F_{q^m}},\F_{q^m})=\PG(k-2,q^m)$ defined by $W'=W+\langle v\rangle_{\F_{q^m}}/\langle v\rangle_{\F_{q^m}} \subseteq V/\langle v\rangle_{\F_{q^m}}$. Since $P \notin \Sigma$, then $L_{W'}$ has rank $k$. Moreover, due to the fact that every line through $P$ meeting $\Sigma$ in at least one point have weight at least $2$, it follows that $w_{L_{W'}}(P')\geq 2$, for each $P' \in L_{W'}$. Moreover, since $L_W$ spans $\PG(k-1,q^m)$ then $L_{W'}$ spans $\PG(k-2,q^m)$ as well. Then by \eqref{eq:boundcomplementaryweight}, we get $k \geq 2(k-1)$, a contradiction. So the statement is true for $r=1$. Suppose now that our assertion holds for $r-1$ and let prove it for $r$. By hypothesis, we have that there exists an $(r-1)$-dimensional space $\Omega$ through $P$ meeting $\Sigma$ in exactly one point $Q$. Suppose that  every $r$-space through $\Omega$ meets $\Sigma$ in at least another point different from $Q$. Then 
\[
    q \frac{q^{m(k-r)}-1 }{q^m-1}+1 \leq \frac{q^k-1}{q-1},
\]
a contradiction since $k \leq m$.
\end{proof}

We will now use the above geometric lemma to prove the case of equality in the lower bound in Theorem \ref{thm:boundn<=m}.

\begin{theorem}\label{thm:charn<mmin}
Let $\C$ be a non-degenerate $[n,k]_{q^m/q}$ code. Assume that $n< m$ and
\[  M(\C)=(q^{mk}-1) -\frac{q^n-1}{q-1}(q^{(k-1)m}-1) +q(q^m-1)\beta, \]
where
\[\beta=\frac{q^{mk}-1}{q^m-1}-\prod_{i=1}^{k-1}(q^m-q^i)-\frac{q^k-1}{q-1}
\prod_{i=1}^{k-2}(q^m-q^i).\]
Then $n=k$ and $\C=\F_{q^m}^k$.
\end{theorem}
\begin{proof}
Let $U$ be any associated system with $\C$. As in the proof of Theorem \ref{thm:boundn<=m}, denote by $\tau_0,\tau_1$ and $\tau_s$ the number of hyperplanes in $\PG(k-1,q^m)$ meeting $L_U$ in $0,1$ and at least two points, respectively, and let $W$ be an $\fq$-subspace of $U$ such that $L_W=\PG(k-1,q)$.
By the assumptions and by \eqref{eq:distancedesign}, 
\[\tau_0=\frac{q^{mk}-1}{q^m-1} -\frac{q^n-1}{q-1}\frac{q^{(k-1)m}-1}{q^m-1} +q\beta,\]
and hence in \eqref{eq:boundtactical} we have equalities
\[\tau_1+\tau_s = |L_U| \frac{q^{(k-1)m}-1}{q^m-1} -q\tau_s = \frac{q^n-1}{q-1}\frac{q^{(k-1)m}-1}{q^m-1} -q\beta.\]
The above equalities imply
\[ \frac{q^{(k-1)m}-1}{q^m-1}\left(|L_U|-\frac{q^n-1}{q-1}\right)-q(\tau_s-\beta)=0, \]
and so $|L_U|=\frac{q^n-1}{q-1}$ and $\beta=\tau_s$.
Suppose that there exists a point $P\in L_U \setminus L_W$. By Lemma \ref{lem:tangenthypersubgeo}, there exists a hyperplane $\pi$ of $\PG(k-1,q^m)$ through $P$ meeting $L_W$ in exactly one point. So, $\pi$ is secant to $L_U$ but tangent to $L_W$, a contradiction to $\tau_s=\beta$. 
Therefore $L_U=L_W$ and hence the assertion.
\end{proof}

\section{Equality in the upper bounds}\label{sec:upperbounds}

The maximum for $M(\C)$ is assumed if and only if either $\C$ or its geometric dual is the entire space.

\begin{theorem}
Let $\C$ be a non-degenerate $[n,k]_{q^m/q}$ code and assume that $d_{k-1}^{\mathrm{rk}}(\C)\geq n-m+1$.
\begin{itemize}
    \item If $n< m$ then $M(\C)$ is maximum with respect to \eqref{eq:boundM(C)n<m} if and only if $n=k$ and $\C=\F_{q^m}^k$.
    \item If $n\geq m$ then $M(\C)$ is maximum with respect to \eqref{eq:boundmC2dualk} if and only if $n=mk-k$ and $\C^{\perp_{\mathcal{G}}}=\F_{q^m}^k$.
\end{itemize} 
\end{theorem}
\begin{proof}
Let $U$ be any system associated with $\C$.
Let start by assuming that $n< m$ and $k=2$, then by Theorem \ref{thm:boundsMC21}, $M(\C)$ is maximum if and only if $L_U$ has size $q+1$ and arguing as before we obtain that this is equivalent to require that $\C$ is $\F_{q^m}^2$. In the case in which $k>2$, then by Theorem \ref{thm:boundn<=m} $M(\C)$ is maximum if and only if $M(\C)=\prod_{i=0}^{n-1}(q^m-q^i)$.
Suppose that $k<n$, then 
\[M(\C)> \prod_{i=0}^{k-1}(q^m-q^i), \]
that is by Remark \ref{rk:exthyperssubgeo} and Proposition \ref{prop:points-codewordsk2} the number of external hyperplanes to $L_U$ is greater than the number of external hyperplanes to $L_W$, where $L_W$ is a canonical subgeometry contained in $L_U$. Since $L_W \subseteq L_U$, this is a contradiction. Hence $k=n$ and we obtain the assertion.
Suppose that $n\geq m$.
By \eqref{eq:boundmC2dualk} and Proposition \ref{prop:points-codewordsk2} we have that $L_{U^{\perp'}}$ has size $\frac{q^k-1}{q-1}$. By \cite[Lemma 3.2]{SamPaolo}, $L_{U^{\perp'}}\simeq\PG(k-1,q)$ is a canonical subgeometry of $\PG(k-1,q^m)$ and $\dim_{\fq}(U^{\perp'})=k$. Hence, $\dim_{\F_{q^m}}(\C^{\perp_{\mathcal{G}}})=k$, i.e. $\C^{\perp_{\mathcal{G}}}= \F_{q^m}^k$. 
\end{proof}

We can also characterize the case of equality in Theorem \ref{th:boundsubgeometrycodes}.

\begin{proposition}
Let $\C$ be a non-degenerate $[n,k]_{q^m/q}$ code and assume that $m\leq n$ and $d_{k-1}^{\mathrm{rk}}(\C)\geq n-m+1$.  Let $G'$ be any of generator matrix of $\C^{\perp_{\mathcal{G}}}$. Suppose there exist $r>1$ codewords $c_1,\ldots,c_r \in \C^{\perp_{\mathcal{G}}}$ $\F_{q^m}$-linearly independent such that
\[ W=\psi_{G'} \left( \bigcap_{i=1}^r \mathrm{supp}(c_i)^\perp \right)\]
satisfies $\dim_{\fq}(W)=\dim_{\F_{q^m}}(\langle W \rangle_{\F_{q^m}})=k-r$
and
\[
M(\C) = q^{km}-1-(q^m-1)\left( q^{km-n}+\ldots+q^{km-n-k+r}+\frac{q^r-1}{q-1}\right).
\]
Then $n=mk-k$ and $\C^{\perp_{\mathcal{G}}}=\F_{q^m}^k$.
\end{proposition}
\begin{proof}
Let $U$ be any system associated with $\C$. Then 
\[ M(\C)= (q^m-1)|\PG(k-1,q^m)\setminus L_{U^{\perp'}}|,\]
and hence
\[|L_{U^{\perp'}}|=q^{km-n-1}+\ldots+q^{km-n-k+r}+\frac{q^r-1}{q-1}.\]
Arguing as in the proof of Theorem \ref{th:boundsubgeometrycodes}, there exists an $\F_{q^m}$-subspace $S$ of $\F_{q^m}^k$ of dimension $k-r<k-1$ such that $L_{W \cap U^{\perp'}}$ is a canonical subgeometry in $\PG(W,\F_{q^m})$ and hence $L_{U^{\perp'}}$ satisfies the assumption of Theorem \ref{thm:newBound} with equality in the lower bound and hence $L_{U^{\perp'}}=\PG(k-1,q)$ by the second part of Theorem \ref{thm:newBound}.
\end{proof}

In the following we will study the case in which, under certain assumptions, the upper bound on $M(\C)$ has been improved.
For this case, the situation is much more complicated and a complete answer in general seems to be very difficult.
Indeed, in this case we will show some examples which will suggest that a complete classification for this case is hard to obtain.

We start by describing some constructions for codes.

\begin{construction}\label{con:polcodek>2}
Let $\lambda\in \F_{q^m} \setminus \fq$ be an element generating $\F_{q^m}$ and 
\[ G=
\left(
\begin{array}{llllllllllll}
 1 & \lambda & \ldots & \lambda^{t_1-1} & 0 & \ldots & & & & & & 0  \\
0 & 0 & \ldots & 0 & 1 & \lambda & \ldots & \lambda^{t_2-1} & 0  & \ldots & &  0\\
\vdots &  &  &  &  &  & &  & \ddots &  & & \\
0 & \ldots &  &  &  &  &  & & 0  & 1 & \ldots &  \lambda^{t_{k}-1}
\end{array}
\right)
\in \F_{q^m}^{k\times (t_1+\ldots+t_k)}. \]
Let $\C_{\lambda,t_1,\ldots,t_{k}}$ be the $\F_{q^m}$-linear rank metric code in $\mathbb{F}_{q^m}^{t_1+\ldots+t_{k}}$ with dimension $k$ having $G$ as a generator matrix. 
\end{construction}

We now determine the parameters of these codes.

\begin{theorem}\label{th:parmsofpolcode}
Let $\lambda\in \F_{q^m} \setminus \fq$ be an element generating $\F_{q^m}$ and let $\C_{\lambda,t_1,\ldots,t_{k}}$ be as in Construction \ref{con:polcodek>2}.
Assume that $t_1\leq t_2\leq \ldots\leq t_k\leq m-1$.
Then $\C_{\lambda,t_1,\ldots,t_{k}}$  is an $[t_1+\ldots+t_k,k,t_1]_{q^m/q}$ code. 
Moreover, if $k=2$ then 
\begin{itemize}
    \item the first row of $G$ in Construction \ref{con:polcodek>2} and its non-zero $\F_{q^m}$-proportional vectors have weight $t_1$;
    \item the number of codewords in $\C_{\lambda,t_1,t_2}$ of weight $t_2$ different from those in the previous item is $q^{t_2-t_1+1}(q^m-1)$;
    \item the number of codewords of weight $\min\{m,n\}-i$ in $\C_{\lambda,t_1,t_2}$ is $(q^m-1)(q^{n-2i+1}-q^{n-2i-1})$, for any $i \in \{1,\ldots,t_1-1\}$.
\end{itemize}
\end{theorem}
\begin{proof}
Because of the structure of $G$, it is clear that the length of $\C_{\lambda,t_1,\ldots,t_{k}}$ is $n=t_1+\ldots+t_k$ and its dimension is $k$.
Let consider the following system associated with $\C_{\lambda,t_1,\ldots,t_k}$
\[U=S_1\times \ldots \times S_k=\langle 1,\lambda,\ldots,\lambda^{t_1-1} \rangle_{\fq}\times \ldots \times  \langle 1,\lambda,\ldots,\lambda^{t_k-1} \rangle_{\fq}. \]
In order to determine the weight distribution of the code we need to determine the weight distribution of $L_U$ with respect to the hyperplanes. To this aim we will consider its dual and we will recover it from the weight distribution of the points of its dual.
Consider 
\[ \sigma \colon ((u_1,\ldots,u_k),(v_1,\ldots,v_k)) \in (\mathbb{F}_{q^m}^k)^2 \mapsto u_1v_1+\ldots+u_kv_k \in \F_{q^m}, \]
in this way 
\[ \sigma' \colon ((u_1,\ldots,u_k),(v_1,\ldots,v_k)) \in (\mathbb{F}_{q^m}^k)^2 \mapsto \mathrm{Tr}_{q^m/q}(u_1v_1+\ldots+u_kv_k) \in \F_{q}. \]
Denote by 
\[\overline{S}_i= \{ a \in \F_{q^m} \colon \mathrm{Tr}_{q^m/q}(ab)=0,\,\,\forall b \in S_i \}, \]
for any $i \in \{1,\ldots,k\}$.
It is easy to see that $U^{\perp'}=\overline{S}_1 \times \ldots\times \overline{S}_k$.
By \cite[Proposition 2.9]{napolitano2022classifications}, there exists $\alpha \in \F_{q^m}^*$ such that
\[ \overline{S}_i=\alpha \langle 1,\lambda,\ldots,\lambda^{m-t_i-1} \rangle_{\fq}, \]
for $i\in \{1,\ldots,k\}$. Therefore, $U^{\perp'}$ is an $\fq$-subspace of dimension $km-n$ in $\F_{q^m}^k$ which is $\mathrm{GL}(k,q^m)$-equivalent to $W$, where
\[W=\langle 1,\lambda,\ldots,\lambda^{m-t_1-1} \rangle_{\fq}\times\ldots\times \langle 1,\lambda,\ldots,\lambda^{m-t_k-1} \rangle_{\fq}. \]
Hence, the weight distributions of $L_{U^{\perp'}}$ and $L_W$ coincide. So, by Remark \ref{rk:maximumweighpointscomplementaryweights}
\[ \max\{w_{L_{U^{\perp'}}}(P) \colon P \in \PG(k-1,q^m)\}=\max\{w_{L_W}(P) \colon P \in \PG(k-1,q^m)\}=m-t_1, \]
therefore by Proposition \ref{prop:weightdual} we have that
\[ \max\{w_{L_{U}}(H) \colon H=\PG(k-2,q^m)\subset\PG(k-1,q^m)\}=m-t_1+n-m=t_2+\ldots+t_k \]
and hence the minimum distance can be determined via Theorem \ref{th:connection}.
When $k=2$ and $n\leq m$, by Theorem \ref{th:connection} the weight distribution of the code can be determined by using the weight distribution of the linear sets in Theorem \ref{th:constructionVdV}. 
If $n\geq m$ we can argue as before with the duality in such a way that the dual of $U$ satisfies the assumptions of Theorem \ref{th:constructionVdV}.
\end{proof}

\begin{remark}
For more general dimensions, it is possible to determine the weight distribution of the code under the assumptions that $t_1+\ldots+t_k\geq m$ and $t_i+t_j \geq m-1$ for any $i\ne j$, by using the duality as in the proof of the above theorem and \cite[Theorem 2.17]{jena2021linear} (see also \cite[Remark 4.3]{SamPaolo}).
\end{remark}

\begin{remark}
The family of codes in Construction \ref{con:polcodek>2} is closed under the operation of geometric dual with respect to $\sigma$ as in the proof of Theorem \ref{th:parmsofpolcode}. More precisely,
$\C_{\lambda,t_1,\ldots,t_k}^{\perp_{\mathcal{G}}}$ is equivalent to $\C_{\lambda,m-t_1,\ldots,m-t_k}$.
\end{remark}

Let's start by proving that the examples of dimension $2$ in Construction \ref{con:polcodek>2} gives the maximum values for $M(\C)$.

\begin{theorem}\label{thm:polcodemax}
Let $\lambda\in \F_{q^m} \setminus \fq$ be an element generating $\F_{q^m}$. The code $\C_{\lambda,t_1,t_2}\subseteq \F_{q^m}^{n}$, where $n=t_1+t_2$ and $3\leq n \leq 2m-3$, has a codeword of weight $\min\{m,n\}-1$ and $\C_{\lambda,t_1,t_2}$ reaches the maximum for $M(\C_{\lambda,t_1,t_2})$ among the $[n,2]_{q^m/q}$ codes with a codeword of weight $\min\{m,n\}-1$.
\end{theorem}
\begin{proof}
Let consider the following system associated with $\C_{\lambda,t_1,t_2}$
\[U=S_1\times S_2=\langle 1,\lambda,\ldots,\lambda^{t_1-1} \rangle_{\fq}\times \langle 1,\lambda,\ldots,\lambda^{t_2-1} \rangle_{\fq}. \]
Suppose that $n\leq m$ then $t_1+t_2\leq m$
and $L_U$ is an $\fq$-linear set of rank $n$ of the form of Theorem \ref{th:constructionVdV}, which implies that $|L_U|=q^{n-1}+1$ and $L_U$ has $q^{n-1}-q^{n-3}\geq 1$ points of weight one. By \eqref{eq:relweight}, the latter fact reads as $\C_{\lambda,t_1,t_2}$ has at least a codeword of weight $n-1$ and by Proposition \ref{prop:points-codewords2}
\[M(\C_{\lambda,t_1,t_2})=q^{2m}-1-(q^m-1)(q^{n-1}+1),\]
that is we have the equality in \eqref{eq:boundmC2scatt1}.
Now, suppose that $n > m$ and consider $U^{\perp'}$ as in the proof of Theorem \ref{th:parmsofpolcode}. Hence, we have that $U^{\perp'}$ is an $\fq$-subspace of dimension $2m-n$ in $\F_{q^m}^2$ which is $\mathrm{GL}(2,q^m)$-equivalent to $W$, where
\[W=\langle 1,\lambda,\ldots,\lambda^{m-t_1-1} \rangle_{\fq}\times \langle 1,\lambda,\ldots,\lambda^{m-t_2-1} \rangle_{\fq}. \]
Hence, $|L_{U^{\perp'}}|=|L_W|$. We can apply again Theorem \ref{th:constructionVdV} to $L_W$ and we obtain that $L_{U^{\perp'}}$ has $q^{2m-n-1}-q^{2m-n-3}\geq 1$ points of weight one and size $q^{2m-n-1}+1$. By combining Proposition \ref{prop:weightdual} and \eqref{eq:relweight}, we have that $\C_{\lambda,t_1,t_2}$ has at least one codeword of weight $m-1$. By \eqref{eq:M(C)dual}, we have that
\[M(\C_{\lambda,t_1,t_2})=q^{2m}-1-(q^m-1)(q^{2m-n-1}+1),\]
that is the equality in \eqref{eq:boundmC2dual} of Theorem \ref{thm:boundMC22}.
\end{proof}

The above result can be extended to larger dimension under certain assumptions.

\begin{theorem}\label{th:boundsuperp}
Let $\lambda\in \F_{q^m} \setminus \fq$ be an element generating $\F_{q^m}$. Let consider the code $\C_{\lambda,t_1,\ldots,t_k}\subseteq \F_{q^m}^{n}$, where $n=t_1+\ldots+t_k$, $m\leq n$ and $km-n\leq m+k$. 
Let $\{m-t_1-1,\ldots,m-t_k-1\}=\{s_{i_1},\ldots,s_{i_{\ell}}\}$, with $s_{i_1}>\ldots > s_{i_{\ell}}$. Then, if either
\[
k \leq \sum_{j=1}^{\ell} \frac{q^{s_{i_j}}-2q^{s_{i_j}/2}}{s_{i_j}}
\]
or  
\[ mk-k-t_1-\ldots-t_k \leq q, \]
the code $\C_{\lambda,t_1,\ldots,t_k}$ satisfies the assumptions of Theorem \ref{th:boundsubgeometrycodes} with $r=1$ and reaches the maximum for $M(\C_{\lambda,t_1,\ldots,t_k})$ among the $[n,k]_{q^m/q}$ codes satisfying the assumptions of Theorem \ref{th:boundsubgeometrycodes} with $r=1$.
\end{theorem}
\begin{proof}
As for the two dimensional case, a system associated with $\C_{\lambda,t_1,\ldots,t_k}$ is 
\[U=S_1\times\ldots\times S_k=\langle 1,\lambda,\ldots,\lambda^{t_1-1} \rangle_{\fq}\times \ldots \times \langle 1,\lambda,\ldots,\lambda^{t_k-1} \rangle_{\fq}. \]
Consider $U^{\perp'}$ as in the proof of Theorem \ref{th:parmsofpolcode}. Hence, we have that $U^{\perp'}$ is an $\fq$-subspace of dimension $km-n$ in $\F_{    ^m}^k$ which is $\mathrm{GL}(k,q^m)$-equivalent to $W$, where
\[W=\langle 1,\lambda,\ldots,\lambda^{m-t_1-1} \rangle_{\fq}\times \ldots \times \langle 1,\lambda,\ldots,\lambda^{m-t_k-1} \rangle_{\fq}. \]
By \cite[Theorem 4.5, Proposition 4.6, Remark 4.7]{SamPaolo}, there exists a hyperplane $H$ such that $L_W\cap H=\PG(k-2,q)$ and hence there exists a codeword in a code equivalent to $\C^{\perp_{\mathcal{G}}}$ as in the statement, and hence as well in $\C^{\perp_{\mathcal{G}}}$. 
\cite[Theorem 4.5]{SamPaolo} and, together with \cite[Proposition 4.6]{SamPaolo} and \cite[Corollary 4.8]{SamPaolo} also implies that
\[|L_W|=|L_{U^{\perp'}}|=q^{km-n-1}+\ldots+q^{km-n-k+1}+1,\]
which implies that, by \eqref{eq:boundmC2scattk2boundnew} with $r=1$, $M(\C_{\lambda,t_1,\ldots,t_k})$ is maximum.
\end{proof}

\begin{remark}
In particular, if $t_1\geq \ldots \geq t_k$ and 
\[k\leq \frac{q^{m-t_1-1}-2q^{\frac{m-t_1-1}2}}{m-t_k-1}\] 
the assumptions of Theorem \ref{th:boundsuperp} are satisfied.
\end{remark}

\begin{remark}
Making use of \cite[Proposition 4.9]{SamPaolo} and assuming that $t_0\leq \ldots \leq t_k$, one can also prove that replacing the assumption $km-n\leq m+k$ by $3m-t_0-t_{k-1}-t_k\leq m+2$, $\C_{\lambda,t_1,\ldots,t_k}$ satisfies the assumptions of Theorem \ref{th:boundsubgeometrycodes} with $r=k-1$ and $M(\C_{\lambda,t_1,\ldots,t_k})$ is the maximum among the codes with this property as well.
\end{remark}

When $m$ is prime and $k=2$, then we can give a characterization of the codes reaching the maximum for $M(\C)$ once we require that there are enough codewords of a certain weight, by making use of the results in \cite{napolitano2022classifications} in which a key role is playled by the linear analogue of the Cauchy-Davenport inequality and Vosper's Theorem, see \cite{bachoc2017analogue,bachoc2018revisiting}. 

\begin{theorem}\label{thm:classminsize}
Let $\C$ be a non-degenerate $[n,2]_{q^m/q}$ code and suppose that $m$ is prime.
Denote by $A_i$ the number of codewords in $\C$ having weight $i \in \{0,\ldots,\min\{m,n\}\}$.
Assume that one of the following holds:
\begin{itemize}
    \item $n\leq m$, there exist two not $\F_{q^m}$-proportional codewords in $\C$ of weight $r$ and $n-r$ (with $r \leq n-r$), respectively, and $A_{n-r}\geq (q^m-1)(q^{n-2r}+2)$;
    \item $n> m$, there exist two not $\F_{q^m}$-proportional codewords in $\C$ of weight $m+r-n$ and $m-r$ (with $r \leq n-r$), respectively, and $A_{m-r}\geq (q^m-1)(q^{2m-n-2r}+2)$.
\end{itemize}
Then $\C$ is equivalent to the code in Construction \ref{con:polcodek>2} and hence it reaches the maximum value for $M(\C)$ among the non-degenerate $[n,2]_{q^m/q}$ codes having at least one codeword of weight $\min\{m,n\}-1$.
\end{theorem}
\begin{proof}
Let $U$ be any system associated with $\C$.
Suppose that $n\leq m$, then by \eqref{eq:relweight} there exists two distinct points $P$ and $Q$ in $\PG(1,q^m)$ such that $w_{L_U}(P)=n-r$ and $w_{L_U}(Q)=r$. Moreover, by the assumption on $A_{n-r}$ and \eqref{eq:relweight}, it follows that the number of points in $L_U$ with weight $r$ is at least
\[ q^{n-2r}+2. \]
Therefore, we can now apply \cite[Theorem 3.12]{napolitano2022classifications} and we have that $U$ is $\mathrm{GL}(2,q^m)$-equivalent to 
\[ \langle 1,\lambda, \ldots, \lambda^{n-r-1} \rangle_{\fq}\times \langle 1,\lambda, \ldots, \lambda^{r-1} \rangle_{\fq}, \]
for some $\lambda \in \F_{q^m}\setminus\fq$, and hence $\C$ is equivalent to $\mathcal{C}_{\lambda,n-r,r}$.
Assume that $n>m$ and consider $U^{\perp'}$.
Because of the assumptions, \eqref{eq:relweight} and Proposition \ref{prop:weightdual}, there exists two distinct points $P$ and $Q$ in $\PG(1,q^m)$ such that $w_{L_{U^{\perp'}}}(P)=n-r$ and $w_{L_{U^{\perp'}}}(Q)=r$, the number of points in $L_{U^{\perp'}}$ with weight $r$ is at least
\[ q^{2m-n-2r}+2. \]
Since the rank of $L_{U^{\perp'}}$ is $2m-n$ we can apply \cite[Theorem 3.12]{napolitano2022classifications} obtaining that $U^{\perp'}$ is $\mathrm{GL}(2,q^m)$-equivalent to 
\[ \langle 1,\lambda, \ldots, \lambda^{n-r-1} \rangle_{\fq}\times \langle 1,\lambda, \ldots, \lambda^{r-1} \rangle_{\fq}, \]
for some $\lambda \in \F_{q^m}\setminus\fq$.
Then $U$ is $\mathrm{GL}(2,q^m)$-equivalent to 
$S\times T$ where
\[ S=\{a \in \F_{q^m} \colon \mathrm{Tr}_{q^m/q}(ab)=0, \forall b \in \langle 1,\lambda, \ldots, \lambda^{n-r-1} \rangle_{\fq}\} \]
and
\[ T=\{a \in \F_{q^m} \colon \mathrm{Tr}_{q^m/q}(ab)=0, \forall b \in \langle 1,\lambda, \ldots, \lambda^{r-1} \rangle_{\fq}\}. \]
By \cite[Corollary 2.7]{napolitano2022linear}, $S\times T$ is $\mathrm{GL}(2,q^m)$-equivalent to
\[ \langle 1,\lambda, \ldots, \lambda^{m-n+r-1} \rangle_{\fq} \times \langle 1,\lambda, \ldots, \lambda^{m-r-1} \rangle_{\fq}. \]
Hence, by Theorem \ref{thm:1-1correspondence} $\C$ is equivalent to $\C_{\lambda,m-n+r,m-r}$.\\
The last part follows by Theorem \ref{thm:polcodemax}.
\end{proof}

When $m$ is not a prime, there are also other non-equivalent examples of codes reaching the maximum value for $M(\C)$ with respect to the upper bound in Theorem \ref{th:boundsubgeometrycodes}.

\begin{construction}\label{con:polcodelifted}
Assume $m=\ell' t$.
Let $\mu\in \F_{q^m}$ such that $\F_{q^t}=\fq(\mu)$, $\overline{S}$ be an $\F_{q^t}$-subspace of $\F_{q^m}$ of dimension $\ell<\ell'$ such that $1 \notin \overline{S}$ 
 and $\overline{S}\cap \F_{q^t}=\{0\}$. 
Let $t_1,\ldots,t_k$ positive integers such that $t_i+t_j \leq t+1$, for each $i \neq j$. For $c_1,\ldots,c_{\ell t}$ an $\fq$-basis of $\overline{S}$, let consider
\[ G=
\left(
\begin{array}{llllllllllllllll}
 c_1 & \ldots & & c_{\ell t} & 1 & \mu & \ldots & \mu^{t_1-1} &0 & \cdots  & & &  &  &  & 0 \\
0 & \ldots & & & \ldots & & \ldots & 0 & 1 & \mu & \ldots & \mu^{t_2-1} & 0  & \ldots & & 0\\
\vdots &  &  & & & &  &  &  & & &  &  \ddots &  & & \\
0 & \ldots &  &  & & & &  &  & &  & & 0  & 1 & \ldots &  \mu^{t_{k}-1}
\end{array}
\right)
\in \F_{q^m}^{k\times n},\]
with $n=\ell t +t_1+\ldots+t_{k}$.
Define $\C_{\overline{S},\mu,t_1,\ldots,t_{k}}$ the $\F_{q^m}$-linear rank metric code in $\mathbb{F}_{q^m}^{n}$ having $G$ as a generator matrix. 
\end{construction}

The parameters of the above construction are the following.

\begin{theorem}\label{th:parmsofliftpolcode}
Let $\C_{\overline{S},\mu,t_1,\ldots,t_k}\subseteq \F_{q^m}^n$ be as in Construction \ref{con:polcodelifted}. 
Then $\C_{\overline{S},\mu,t_1,\ldots,t_k}$ is a non-degenerate $[n,k,t_j]_{q^m/q}$ code, where $t_j=\min\{t_i \colon i>1\}$ and  $\C_{\overline{S},\mu,t_1,\ldots,t_k}^{\perp_{\mathcal{G}}}$ is a non-degenerate $[km-n,k,t(\ell'-\ell)-t_1]_{q^m/q}$ code.

In the case that $k=2$ and $n \leq m$, then the second row of $G$ in Construction \ref{con:polcodelifted} and its non-zero $\F_{q^m}$-proportional vectors are exactly all the codewords of weight $t_2$. And if $t_1 \geq t_2$ then
\begin{itemize}
    \item the number of codewords in $\C_{\overline{S},\mu,t_1,\ldots,t_k}$ of weight $n-t_2$ is $q^{\ell t+t_1-t_2+1}(q^m-1)$;
    \item the number of codewords of weight $n-i$ in $\C_{\overline{S},\mu,t_1,\ldots,t_k}$ is $(q^m-1)(q^{n-2i+1}-q^{n-2i-1})$, for any $i \in \{1,\ldots,t_2-1\}$.
\end{itemize}
If $t_1 < t_2$ then
\begin{itemize}
    \item the number of codewords in $\C_{\overline{S},\mu,t_1,\ldots,t_k}$ of weight $\ell t +t_1$ is $q^{\ell t}(q^m-1)$;
    \item the number of codewords in $\C_{\overline{S},\mu,t_1,\ldots,t_k}$ of weight $\ell t +t_2$ is $(q^{\ell t+t_2-t_1+1}-q^{\ell t})(q^m-1)$;
    \item the number of codewords of weight $n-i$ in $\C_{\overline{S},\mu,t_1,\ldots,t_k}$ is $(q^m-1)(q^{n-2i+1}-q^{n-2i-1})$, for any $i \in \{1,\ldots,t_1-1\}$.
\end{itemize}
\end{theorem}
\begin{proof}
A system associated with $\C_{\overline{S},\mu,t_1,\ldots,t_k}$ is 
\[U=(\overline{S} \oplus \langle 1,\mu,\ldots,\mu^{t_1-1} \rangle_{\fq} ) \times \ldots \times  \langle 1,\mu,\ldots,\mu^{t_k-1} \rangle_{\fq}. \]
Let $U_1=\overline{S} \oplus \langle 1,\mu,\ldots,\mu^{t_1-1} \rangle_{\fq}$ and $U_i=\langle 1,\mu,\ldots,\mu^{t_i-1} \rangle_{\fq}$, for $i>1$.
As in Theorem \ref{th:parmsofpolcode}, consider 
\[ \sigma \colon ((u_1,\ldots,u_k),(v_1,\ldots,v_k)) \in (\mathbb{F}_{q^m}^k)^2 \mapsto u_1v_1+\ldots+u_kv_k \in \F_{q^m}, \]
and so
\[ \sigma' \colon ((u_1,\ldots,u_k),(v_1,\ldots,v_k)) \in (\mathbb{F}_{q^m}^k)^2 \mapsto \mathrm{Tr}_{q^m/q}(u_1v_1+\ldots+u_kv_k) \in \F_{q}. \]
Denote by 
\[\overline{U}_i= \{ a \in \F_{q^m} \colon \mathrm{Tr}_{q^m/q}(ab)=0,\,\,\forall b \in U_i \}, \]
for any $i \in \{1,\ldots,k\}$.
Clearly $\dim_{\F_q}(\overline{U}_1)=m-\ell t-t_1$, $\dim_{\F_q}(\overline{U}_i)=m-t_i$, for each $i>1$. Moreover, it is easy to see that $U^{\perp'}=\overline{U}_1 \times \ldots\times \overline{U}_k$.
Hence, by Proposition \ref{prop:weightdual}, we get
\[ \max\{w_{L_{U}}(H) \colon H=\PG(k-2,q^m)\subset\PG(k-1,q^m)\}=\]\[n-m+\max\{w_{L_{U^{\perp'}}}(P) \colon P \in \PG(k-1,q^m)\}=n-t_j, \]
where $t_j=\min\{t_i \colon i>1\}$.
So, $d=n-(n-t_j)=t_j$. This implies that  $\C_{\overline{S},\mu,t_1,\ldots,t_k}$ is a non-degenerate $[n,k,t_j]_{q^m/q}$ code. Now, let consider $\C_{\overline{S},\mu,t_1,\ldots,t_k}^{\perp_{\mathcal{G}}}$. A system associated with $\C_{\overline{S},\mu,t_1,\ldots,t_k}^{\perp_{\mathcal{G}}}$ is $U^{\perp'}$, and so it is a $[km-n,k]_{q^m/q}$ code.
Moreover, again by Proposition \ref{prop:weightdual}, we have
\[ \max\{w_{L_{U^{\perp'}}}(H) \colon H=\PG(k-2,q^m)\subset\PG(k-1,q^m)\}\]
\[=(k-1)m-n+\max\{w_{L_U}(P) \colon P \in \PG(k-1,q^m)\}=
(k-1)m-n+\ell t+t_1, \]
implying that $d(\C_{\overline{S},\mu,t_1,\ldots,t_k})=t(\ell'-\ell)-t_1$.
When $k=2$ and $n\leq m$, by Theorem \ref{th:connection} the weight distribution of the code can be determined by using the weight distribution of the linear set defined by $U$, that is completely determined in \cite[Corollary 4.2]{napolitano2022classifications}. 
\end{proof}

We show that the above described construction still yields a code with the maximum possible value for $M(\C)$.

\begin{proposition}
Let $\C_{\overline{S},\mu,t_1,t_2}\subseteq \F_{q^m}^n$ be as in Construction \ref{con:polcodelifted} and assume that $n \leq m$. Then
$\C_{\overline{S},\mu,t_1,t_2}$ has a codeword of weight $n-1$ and $\C_{\overline{S},\mu,t_1,t_2}$ reaches the maximum value for $M(\C)$ among the non-degenerate $[n,2]_{q^m/q}$ codes having at least one codeword of weight $n-1$.
\end{proposition}
\begin{proof}
Let consider the following system associated with $\C_{\overline{S},\mu,t_1,t_2}$
\[U=(\overline{S}\oplus \langle 1,\mu,\ldots,\mu^{t_1-1}\rangle_{\fq})\times \langle 1,\mu,\ldots,\mu^{t_2-1}\rangle_{\fq}. \]
Note that $L_U$ is an $\fq$-linear set of rank $n$ of the form of \cite[Corollary 4.2]{napolitano2022classifications}, which implies that $|L_U|=q^{n-1}+1$ and $L_U$ has $q^{n-1}-q^{n-3}\geq 1$ points of weight one. By \eqref{eq:relweight}, the latter fact reads as $\C_{\overline{S},\mu,t_1,t_2}$ has at least a codeword of weight $n-1$ and by Proposition \ref{prop:points-codewords2}, we have the equality in \eqref{eq:boundmC2scatt1}.
\end{proof}

\begin{remark}
Let $U$ and $U'$ be two systems associated with Construction \ref{con:polcodek>2} and Construction \ref{con:polcodelifted}, respectively. By \cite[Theorem 5.1]{napolitano2022classifications}, $U$ and $U'$ are $\Gamma \mathrm{L}(2,q^m)$-inequivalent and hence by Theorem \ref{thm:1-1correspondence} the codes $\C$ and $\C'$ are inequivalent.
\end{remark}

We now show more examples of rank metric codes satisfying the equality in the upper bound in \ref{th:boundsubgeometrycodes}.

\begin{theorem}
Let $\C_{\overline{S},\mu,t_1,\ldots,t_k}\subseteq \F_{q^m}^n$ be as in Construction \ref{con:polcodelifted} and assume that $t_1+\ldots+t_k \leq t+k$. 
Let $\{t_1-1,\ldots,t_k-1\}=\{s_{i_1},\ldots,s_{i_{\ell}}\}$, with $s_{i_1}>\ldots > s_{i_{\ell}}$. If either
\[
k \leq \sum_{j=1}^{\ell} \frac{q^{s_{i_j}}-2q^{s_{i_j}/2}}{s_{i_j}}
\]
or
\[ t_1+\ldots+t_k-k \leq q, \]
then the code $\C^{\perp_{\mathcal{G}}}_{\overline{S},\mu,t_1,\ldots,t_k}\subseteq \F_{q^m}^{km-n}$ satisfies the assumptions of Theorem \ref{th:boundsubgeometrycodes} with $r=1$ and $\C^{\perp_{\mathcal{G}}}_{\overline{S},\mu,t_1,\ldots,t_k}$  reaches the maximum value for $M(\C)$ among the codes satisfying the assumption of Theorem \ref{th:boundsubgeometrycodes} with $r=1$.
\end{theorem}
\begin{proof}
Note that $(\C^{\perp_{\mathcal{G}}}_{\overline{S},\mu,t_1,\ldots,t_k})^{\perp_{\mathcal{G}}}$ is equivalent to the code $\C_{\overline{S},\mu,t_1,\ldots,t_k}$ and an associated system is 
\[U=(\overline{S} \oplus \langle 1,\mu,\ldots,\mu^{t_1-1} \rangle_{\fq} ) \times \ldots \times  \langle 1,\mu,\ldots,\mu^{t_k-1} \rangle_{\fq}. \]
Since $t_1+\ldots+t_k \leq t+k$, by \cite[Corollary 4.16]{SamPaolo}, there exists a hyperplane $H$ such that $L_{U} \cap H=\PG(k-2,q)$ and hence there exists a codeword in $\C^{\perp_{\mathcal{G}}}_{\overline{S},\mu,t_1,\ldots,t_k}$ as in the statement of Theorem \ref{th:boundsubgeometrycodes}, and hence as well in $(\C^{\perp_{\mathcal{G}}})^{\perp_{\mathcal{G}}}$. Moreover, 
 \cite[Theorem 4.14]{SamPaolo} also implies that
\[|L_U|=q^{n-1}+\ldots+q^{n-k+1}+1,\]
which implies that $M(\C^{\perp_{\mathcal{G}}}_{\overline{S},\mu,t_1,\ldots,t_k})$ is maximum.
\end{proof}

\begin{remark}
In particular, if $t_1\geq \ldots \geq t_k$ and $k\leq \frac{q^{t_k-1}-2q^{\frac{t_k-1}2}}{t_1-1}$ the assumptions of Theorem \ref{th:boundsuperp} are satisfied.
\end{remark}

\section{Conclusions and open problems}

In this paper we provide upper and lower bounds on the number of codewords of an $\F_{q^m}$-linear non-degenerate rank metric with maximum weight.
The upper bounds have been improved under certain assumptions.
Then we gave some characterization results, even if in some cases we do not know if the obtained bounds are sharp.

Here we list some open problems, which may be of interest for the reader.

\begin{itemize} 
    \item We do not know whether or not the lower bound in Theorem \ref{thm:boundn<=m2} is sharp. So, it would be interesting to construct examples satisfying the equality or to improve this bound. 
    \item To extend the characterization of Theorem \ref{thm:classminsize} to larger dimension. 
    \item What about the density of the codes having extreme values for $M(\C)$? It would be interesting to see whether the techniques developed in \cite{antrobus2019maximal,gruica2022common} could be adapted also in this case.
    \item What is the average value of $M(\C)$?
\end{itemize}

\bibliographystyle{abbrv}
\bibliography{biblio}

Olga Polverino, Paolo Santonastaso and Ferdinando Zullo\\
Dipartimento di Matematica e Fisica,\\
Universit\`a degli Studi della Campania ``Luigi Vanvitelli'',\\
I--\,81100 Caserta, Italy\\
{{\em \{olga.polverino,paolo.santonastaso,ferdinando.zullo\}@unicampania.it}}

\end{document}